\let\accentvec\vec  
\documentclass[runningheads]{llncs}
\let\vec\accentvec 
\usepackage[utf8x]{inputenc}

\usepackage{amsmath,amssymb}
\usepackage{xspace}
\usepackage{graphicx}
\usepackage{tikz}

\usepackage{algorithm}
\usepackage[noend]{algorithmic}
\usepackage{algorithmic-fix}
\newenvironment{claimproof}{\begin{proof}}{\hfill\qed\end{proof}}

\renewcommand{\algorithmicrequire}{\textbf{Input:}}
\renewcommand{\algorithmicensure}{\textbf{Output:}}

\graphicspath{{./images/}}

\usepackage{bibunits}

\usepackage[pdftex,                
            bookmarks=false,       
            bookmarksnumbered=false,
            bookmarksopen=false,
]{hyperref}
\makeatletter
\providecommand*{\toclevel@title}{0}
\providecommand*{\toclevel@author}{0}
\makeatother

\newcommand{\drawvertex}[1]{\draw[fill=white] #1 circle (0.07);}
\newcommand{\drawblackvertex}[1]{\draw[fill=black] #1 circle (0.1);}
\newcommand{\drawemptyvertex}[1]{\draw[fill=white] #1 circle (0.1);}
\newcommand{\drawgreyvertex}[1]{\draw[fill=white!50!black] #1 circle (0.1);}

\newcommand{\introduceparameterizedproblem}[4]{
\begin{quote}
\noindent\textbf{#1}

\noindent\textbf{Input:} #2

\noindent\textbf{Parameter:} #3

\noindent\textbf{Question:} #4
\end{quote}}

\newcommand{\KB}{\ensuremath{B_{K_4}}\xspace}

\newcommand{\yes}[0]{\textsc{yes}\xspace}
\newcommand{\no}[0]{\textsc{no}\xspace}

\newcommand{\collapse}[0]{PH~$= \Sigma _3 ^p$\xspace}
\newcommand{\containment}[0]{NP~$\subseteq$~coNP$/$poly\xspace}
\newcommand{\ncontainment}[0]{NP~$\not \subseteq$~coNP$/$poly\xspace}

\newcommand{\K}[0]{\ensuremath{\mathcal{K}}\xspace}
\newcommand{\B}[0]{\ensuremath{\mathcal{B}}\xspace}
\newcommand{\I}[0]{\ensuremath{\mathcal{I}}\xspace}

\renewcommand{\H}[0]{\ensuremath{\mathcal{H}}\xspace}

\newcommand{\R}[0]{\ensuremath{\mathcal{R}}\xspace}

\renewcommand{\P}{\ensuremath{\mathcal{P}}\xspace}
\newcommand{\T}{\ensuremath{\mathcal{T}}\xspace}
\renewcommand{\L}[0]{\ensuremath{\mathcal{L}}\xspace}

\newcommand{\C}[0]{\ensuremath{\mathcal{C}}\xspace}
\newcommand{\G}[0]{\ensuremath{\mathcal{G}}\xspace}
\newcommand{\F}[0]{\ensuremath{\mathcal{F}}\xspace}

\newcommand{\tw}[0]{\mathop{\mathrm{\textsc{tw}}}}
\newcommand{\opt}[0]{{\mathop{\mathrm{\textsc{opt}}}}}

\newcommand{\mbound}[0]{\ensuremath{m(m+3)/2}\xspace}

\newlength{\baseImageHeight}
\newcommand{\octproblem}[0]{\textsc{Odd Cycle Transversal}\xspace}
\newcommand{\weightedoctproblem}[0]{\textsc{Weighted Odd Cycle Transversal}\xspace}

\newcommand{\oct}[0]{\textsc{oct}\xspace}
\newcommand{\loct}[0]{$\ell$-\textsc{oct}\xspace}

\newcommand{\forest}[0]{\textsc{forest}\xspace}
\newcommand{\cluster}[0]{\textsc{cluster}\xspace}
\newcommand{\cocluster}[0]{\textsc{co-cluster}\xspace}
\newcommand{\bipartite}[0]{\textsc{bip}\xspace}
\newcommand{\outerplanar}[0]{\textsc{outerplanar}\xspace}
\newcommand{\edgeless}[0]{\textsc{edgeless}\xspace}
\newcommand{\gtw}[1]{\G_{\tw({#1})}\xspace}
\newcommand{\biptwoct}[0]{$(\bipartite \cap \gtw{w})$-\oct}
\newcommand{\annotatedbiptwoct}[0]{\textsc{Annotated} $(\bipartite \cap \gtw{w})$-\oct}
\newcommand{\restrictedbiptwoct}[0]{\textsc{Restricted Annotated} $(\bipartite \cap \gtw{w})$-\oct}

\spnewtheorem*{theoremstar}{Theorem}{\bfseries}{\itshape}

\newcommand{\sectref}[1]{Section~\ref{#1}}

\newcommand{\defref}[1]{Definition~\ref{#1}}
\newcommand{\lemmaref}[1]{Lemma~\ref{#1}}

\newcommand{\thmref}[1]{Theorem~\ref{#1}}

\newcommand{\imgref}[1]{Figure~\ref{#1}}

\newcommand{\proposref}[1]{Proposition~\ref{#1}}

\newcommand{\BigO}{\mathcal{O}}
\newcommand{\middlemid}{\,\middle\vert\,}

\usepackage{xcolor}

\usepackage{algorithm}

\titlerunning{On Polynomial Kernels for Structural Parameterizations of OCT}
\title{On Polynomial Kernels for Structural Parameterizations of Odd Cycle Transversal\thanks{This work was supported by the Netherlands Organization for Scientific Research (NWO), project ``KERNELS: Combinatorial Analysis of Data Reduction''.}}

\author{Bart M. P. Jansen and Stefan Kratsch}

\institute{
Utrecht University, The Netherlands,
\email{$\{$bart,kratsch$\}$@cs.uu.nl}
}

\begin{document}

\maketitle

\begin{abstract}
The \octproblem problem (\oct) asks whether a given graph can be made bipartite (i.e.,~$2$-colorable) by deleting at most~$\ell$ vertices.
We study structural parameterizations of \oct with respect to their polynomial kernelizability, i.e., whether instances can be efficiently reduced to a size polynomial in the chosen parameter. It is a major open problem in parameterized complexity whether \octproblem admits a polynomial kernel when parameterized by~$\ell$.

On the positive side, we show a polynomial kernel for \oct when parameterized by the vertex deletion distance to the class of bipartite graphs of treewidth at most~$w$ (for any constant~$w$); this generalizes the parameter feedback vertex set number (i.e., the distance to a forest).

Complementing this, we exclude polynomial kernels for \oct parameterized by the distance to outerplanar graphs, conditioned on the assumption that~NP~$\nsubseteq$~coNP/poly. Thus the bipartiteness requirement for the treewidth~$w$ graphs is necessary. Further lower bounds are given for parameterization by distance from cluster and co-cluster graphs respectively, as well as for \textsc{Weighted} \oct parameterized by the vertex cover number (i.e., the distance from an independent set).
\end{abstract}

\begin{bibunit}[abbrv]

\section{Introduction}
\textsc{Odd Cycle Transversal} (\oct), also called \textsc{Graph Bipartization}, is the task of making an undirected graph bipartite by deleting as few vertices as possible; such a set is a transversal of the odd-length cycles in the graph. The \oct problem has applications in computational biology~\cite{RizziBIL02,Wernicke03}, amongst others. It is NP-complete and admits a polynomial-time~$\BigO(\log n)$-factor approximation algorithm~\cite{GargVY94}; no constant-factor approximation is possible unless Khot's Unique Games Conjecture fails~\cite{Khot02,Wernicke03}. 

In this work we study the parameterized complexity~\cite{DowneyF99} of \oct, focusing on data reduction and kernelization. Parameterized analysis measures the complexity of an algorithm in two dimensions, the input size~$|x|$ and an additional \emph{parameter}~$k \in \mathbb{N}$ which expresses some property of the instance, such as the size of the desired solution. A parameterized problem is a language~$Q \subseteq \Sigma ^* \times \mathbb{N}$, and~$Q$ is (strictly uniformly) \emph{fixed-parameter tractable} (FPT) if there is an algorithm that decides whether~$(x, k) \in Q$ with running time bounded by~$f(k) |x|^{\BigO(1)}$ for some computable function~$f$. 

In the standard parameterization of \octproblem which we call \loct, the parameter~$k := \ell$ measures the number of allowed vertex deletions~$\ell$: an instance is a tuple~$((G,\ell), k := \ell)$ where~$G$ is a graph and~$\ell \in \mathbb{N}$, and the question is whether there is a set~$S \subseteq V(G)$ of size at most~$\ell$ such that~$G - S$ is bipartite. The \loct problem has been very important to the development of parameterized algorithmics, since the algorithm given by Reed, Smith and Vetta~\cite{ReedSV04} to solve \loct in~$\BigO(4^{\ell} \ell mn)$ time\footnote{H\"uffner~\cite{Huffner09} re-analyzed the algorithm and showed it has time complexity~$\BigO(3^{\ell} \ell mn)$.} introduced the technique of \emph{iterative compression} which has turned out to be a key ingredient in finding FPT algorithms for  \textsc{Directed Feedback Vertex Set}~\cite{ChenLLOR08} and \textsc{Multicut}~\cite{BousqetDT11,MarxR10}, amongst others. There has been a significant amount of work on improved exact and parameterized algorithms for \oct and related problems~\cite{RamanSS05,GuoGHNW06,FioriniHRV08,Huffner09,KawarabayashiR10,MarxOR10}.

Kernelization is an important subfield of parameterized complexity which studies polynomial-time preprocessing~\cite{GuoN07a}. A \emph{kernelization algorithm} (or \emph{kernel}) for a parameterized problem~$Q$ is a polynomial-time algorithm which transforms an input~$(x,k) \in \Sigma^* \times \mathbb{N}$ into an \emph{equivalent} reduced instance~$(x', k')$ such that~$|x'|, k' \leq f(k)$ for some computable function~$f$, which is called the \emph{size} of the kernel. All problems in FPT admit kernels for some suitable function~$f$, but \emph{polynomial kernels} (where~$f(k) \in k^{\BigO(1)}$) are of particular interest. It is a famous open problem whether or not \loct admits a polynomial kernel~\cite{Huffner09,GuoGHNW06}. At the 2010 workshop on kernelization WORKER, this was stated as one of the two main open problems in kernelization to date. Even finding a polynomial kernel for \loct restricted to planar graphs was listed as an open problem by Bodlaender et al.\ in the full version of their work~\cite{BodlaenderFLPST09}, despite the fact that planarity makes it significantly easier to obtain polynomial kernels. 

\textbf{Our contribution. }
We study the existence of polynomial kernels for various \emph{structural} parameterizations of the \oct problem. While we have not been able to settle the question of whether \loct admits a polynomial kernel, we do give several upper- and lower bound results for kernel sizes that we believe are important steps towards resolving the main problem. All parameterized problems we consider fit into the following scheme, where \F is a class of graphs:
\introduceparameterizedproblem
{Odd Cycle Transversal parameterized by vertex-deletion distance to~$\boldsymbol{\F}$ [($\boldsymbol{\F}$)-OCT]}
{A graph~$G$, an integer~$\ell$ and a set~$X$ such that~$G - X \in \F$.}
{$k := |X|$.}
{Is there a set~$S \subseteq V(G)$ of size at most~$\ell$ such that~$G - S$ is bipartite?}
We give kernelization upper- and lower bounds for such parameterized problems.

\emph{Upper bounds.}
Our initial goal was to study \oct parameterized by the size of a feedback vertex set (FVS) of the input graph. Recall that a FVS can be defined as a set of vertices whose deletion turns the graph into a forest, and hence this is the (\forest)-\oct problem. After having obtained a polynomial kernel for this problem, we considered generalizations and were able to extend our result significantly. Let \bipartite denote the class of all bipartite graphs, and let~$\gtw{w}$ denote the graphs of treewidth at most~$w$. It is well-known that~$\forest = \bipartite \cap \gtw{1}$. We extended our result for feedback vertex number by showing that for every constant~$w$, the problem ($\bipartite \cap \gtw{w}$)-\oct has a polynomial kernel. Using an approximation algorithm to compute the set~$X$ we can even drop the requirement that the set~$X$ is given in the input; the size of the reduced instance will then be bounded polynomially in the minimum-size of such a set~$X$. Our result can therefore be stated 
as follows: for every fixed~$w\geq 1$ there is a polynomial-time algorithm that transforms an instance~$(G, \ell)$ of \oct into an equivalent instance whose size is bounded by a polynomial in~$|X|$, where~$X \subseteq V(G)$ is a smallest vertex set such that~$G - X \in \bipartite \cap \gtw{w}$.

We believe that the ingredients of our kernelization will be useful for solving the main open problem of whether \loct admits a polynomial kernel. Our kernel uses several powerful techniques from the area of parameterized algorithmics; here is a brief overview. We introduce an annotated version of the problem and show that using these annotations the problem essentially reduces to a connectivity problem with respect to the vertices of the deletion set~$X$. We give a lemma which shows that the main structure of the problem instance lies within an $|X|^{\BigO(1)}$-sized set of connected components of the graph~$G - X$. Using a technique originating in the study of protrusion-based kernelization~\cite{BodlaenderFLPST09} we show the fact that~$G - X \in \gtw{w}$ implies that the number of vertices from~$V(G) \setminus X$ on the boundary of such regions can be bounded by a constant. We analyze the structure of a solution inside such a region in terms of combinatorial properties of separators in labeled graphs. Using the concept of \emph{important separators} as introduced by Marx~\cite{Marx06c} we prove that the number of separators which are relevant to the problem can be bounded polynomially in~$|X|$. To obtain the polynomial kernel we then show how to get rid of vertices which do not belong to any relevant separator.

\emph{Lower bounds.} As described in the previous paragraph, we show the existence of polynomial kernels for~$(\bipartite \cap \gtw{w})$-\oct. We can also prove that the bipartiteness condition cannot be dropped (under a reasonable complexity-theoretic assumption). Observe that~$\gtw{1}$ coincides with the class of forests, and hence only contains bipartite graphs. But~$\gtw{2}$ is the first class of bounded-treewidth graphs which contains non-bipartite graphs, and we prove using cross-composition~\cite{BodlaenderJK11} that~$(\gtw{2})$-\oct does not admit a polynomial kernel unless \containment, which implies a collapse of the polynomial-time hierarchy to the third level (\collapse) and further. We actually prove that~$(\outerplanar)$-\oct does not admit a polynomial kernel under this assumption, which is a stronger statement since~$\outerplanar \subseteq \gtw{2}$. We also show that if we take~\F to be a class of non-bipartite but very simply structured graphs (such as cluster graphs, the union of cliques, or their edge-complements co-cluster graphs) then we cannot obtain polynomial kernels:~$(\cluster)$-\oct and~$(\cocluster)$-\oct do not admit polynomial kernels unless \containment. Since (co)cluster graphs have a very limited structure, the vertex-deletion distance to these graph classes will often be very large; our results show that even for such a large parameter one should not expect to find a polynomial kernel. Finally we look at the vertex-weighted version of \oct and prove that in the presence of vertex weights we cannot even obtain a polynomial kernel measured by the vertex deletion distance to an edgeless graph:~$(\edgeless)$-\weightedoctproblem (which is equivalent to \weightedoctproblem parameterized by the cardinality of a vertex cover) does not admit a polynomial kernel unless \containment. All parameterizations for which we prove kernel lower bounds can be seen to be fixed-parameter tractable because the classes \F have bounded cliquewidth~\cite{CourcelleMR00} and therefore the cliquewidth of the input graphs is bounded by a function of the parameter. 

\textbf{Related work.}
Recent work of Kratsch and Wahlstr\"om~\cite{KratschW11} gives a randomized polynomial kernel for \loct, using matroid theory. To the best of our knowledge no deterministic (and combinatorial) polynomial kernel is known for \loct or for any non-trivial parameterizations of the \oct problem. Wernicke~\cite{Wernicke03} used several reduction rules for \oct as part of his branch-and-bound algorithm, but these rules were not analyzed within the framework of kernelization and do not give provable bounds on the size of reduced instances with respect to any graph parameter. Kernelization with respect to structural parameterizations has been studied by a handful of authors, e.g.,~\cite{FellowsLMMRS09,BodlaenderJK11,JansenB11,UhlmannW10}. 

\textbf{Organization.} We start by giving some preliminaries. In \sectref{section:separators} we give combinatorial bounds for separators in labeled graphs, which will be used in the kernelization algorithm. \sectref{section:kernelization} presents the polynomial kernel for~$(\bipartite \cap \gtw{w})$-\oct. We briefly discuss the kernelization lower bounds in \sectref{section:lowerbounds} and conclude in \sectref{section:conclusion}.

\section{Preliminaries} \label{section:preliminaries}
All graphs considered in this work are simple, undirected, and finite. If~$G$ is a graph then~$V(G)$ and~$E(G)$ denote the vertex- and edge set, respectively. We let \emph{length} and \emph{parity of a path} refer to the number of its vertices. For a vertex~$v \in V(G)$ the open neighborhood is denoted by~$N_G(v)$ and the closed neighborhood is~$N_G[v] := N_G(v) \cup \{v\}$. The open neighborhood of a set~$S \subseteq V(G)$ is~$N_G(S) := \bigcup _{v \in S} N_G[v] \setminus S$. The graph~$G - S$ is the result of removing all vertices in~$S$ and their incident edges from~$G$. We use~$[n]$ as a shorthand for~$\{1, \ldots, n\}$. The term~$\binom{X}{n}$ denotes the collection of all size-$n$ subsets of the finite set~$X$, whereas~$\binom{X}{\leq n}$ represents the collection of size \emph{at most}~$n$ subsets of~$X$. The sizes of these collections are denoted by~$\binom{|X|}{n}$ and~$\binom{|X|}{\leq n}$, respectively.

\section{Combinatorial properties of separators in labeled graphs} \label{section:separators}
An important part of our kernelization relies on a combinatorial bound on the number of essentially distinct ways to separate terminals from labeled vertices in a graph: we prove that if the number of terminals and the size of the separators is taken as a constant, then the number of distinct ways to separate the labels grows polynomially with the number of labels. We believe this to be of independent interest. Some definitions are needed to formalize these claims.
\begin{definition} \label{reachableLabelsDef}
A \emph{labeled graph} is a tuple~$(G, L, f)$ where~$G$ is a graph,~$L$ is a finite set of labels, and~$f \colon V(G) \to 2^L$ is a labeling function which assigns to each vertex a (possibly empty) subset of the labels. For a subset~$S \subseteq V(G)$ and terminal~$t \in V(G)$ we denote by~$R(t, S)$ the \emph{vertices} of~$V(G) \setminus S$ reachable from~$t$ in~$G - S$. The \emph{labels} reachable from~$t$ in~$G - S$ are~$\L(t, S) := \bigcup _{v \in R(t, S)} f(v)$.
\end{definition}

\begin{definition} \label{cutCharacteristicDef}
Let~$(G, L, f)$ be a labeled graph and let~$T = t_1, \ldots, t_n$ be a sequence of distinct terminal vertices in~$G$. The \emph{cut characteristic}~$\K(S, T)$ of a set~$S \subseteq V(G)$ with respect to the terminals~$T$ is an~$n$-dimensional vector $\K(S, T) := ( \L(t_1, S), \L(t_2, S), \ldots, \L(t_n, S))$ whose elements are subsets of~$L$. The set of \emph{distinct cut characteristics}~$\K^m(T)$ for separators of size at most~$m \geq 1$ is~$\K^m(T) := \left \{ \K(S, T)  \middlemid S \in \binom{V(G)}{\leq m} \right \}$.
\end{definition}

Marx~\cite{Marx06c} introduced the notion of \emph{important separators}, and proved their number to be bounded, independently of the graph size. An involved argument which relates important separators to distinct cut characteristics yields the following theorem.

\begin{theorem} \label{simpleCutCharacteristicBound}
Let~$\kappa(n, m, r)$ denote the maximum of~$|\K^m(T)|$ over all labeled graphs~$(G, L, f)$ with~$|L| \leq r$ and over all sets of terminals~$T = \{v_1, v_2, \ldots, v_n\} \subseteq V(G)$, i.e., the maximum number of distinct cut characteristics induced by~$m$-vertex separators in an~$n$-terminal graph labeled with~$r$ different labels. Then $\kappa(n, m, r) \in \BigO(m^{2n} \cdot r^{nm(m+3)/2} \cdot 4^{nm})$, which is polynomial in~$r$ for fixed~$n, m$.
\end{theorem}


\section{Polynomial kernelization for (BIP~$\cap$~\rm{$\gtw{w}$})\bf{-OCT}}\label{section:kernelization}

In this section we describe our polynomial kernelization for \biptwoct. 
Note that the definition of \biptwoct assumes a deletion set to be given in the input, and our kernelization will relate to its size. We will discuss the approximability of the deletion set at the end of the section, which will extend our kernelization to the case that no deletion set is given.

To simplify the formulation of the reduction process, we will actually work with an annotated version of the problem. To obtain the final reduced instance we will later undo these annotations at a small cost.

\introduceparameterizedproblem
{Annotated (BIP~$\boldsymbol{\cap}$~\rm{$\boldsymbol{\gtw{w}}$})\bf{-OCT}}
{A graph~$G$, a set~$X\subseteq V(G)$ such that~$G - X \in (\bipartite \cap \gtw{w})$, a set~$M\subseteq\binom{X}{2}$, and an integer~$\ell$.}
{$k:=|X|$.}
{Is there a set~$S \subseteq V(G)$ of size at most~$\ell$ such that~$G - S$ is bipartite, and there is a proper~$2$-coloring~$c$ of~$G - S$ such that~$c(p)=c(q)$ for all~$\{p,q\}\in M$?}

We call vertex pairs~$\{p,q\}\in M$ \emph{monochromatic}, and these annotations allow us to easily talk about vertices which are constrained to have the same color in~$G - S$. Observe that the dual notion, vertices~$p, q \in X$ which must receive different colors in the bipartite graph~$G - S$, is expressed simply through the existence of an edge~$\{p,q\}$. We will therefore refer to vertices~$p,q \in X$ which are adjacent as vertices to be annotated as \emph{bichromatic}. 
There is no reason a priori that a pair~$\{p,q\}$ cannot be constrained to be simultaneously bichromatic and monochromatic; this condition implies that any valid solution has to delete at least one vertex of the pair before a proper coloring~$2$-coloring can be found. A coloring is said to \emph{respect all annotations} if it respects all edges between vertices of~$X$ as well as the monochromatic pairs given by the set~$M$.

The following straightforward lemma will be used in a number of proofs throughout this section. It shows that any partial~$2$-coloring of a graph whose uncolored parts are bipartite can either be extended to a~$2$-coloring of the whole graph, or one finds a path between two already colored vertices whose parity does not match their colors (e.g., the path has an odd number of internal vertices but the color of the endpoints is different).

\begin{lemma}\label{lemma:twocloringextension}
Let~$G$ be a graph, let~$S\subseteq V(G)$ be such that~$G-S$ is bipartite, and let~$c \colon S \to \{0,1\}$ be a proper~$2$-coloring of~$G[S]$. Then in polynomial time one finds either an extension of~$c$ to a proper~$2$-coloring of~$G$, or a connected component~$C$ of~$G-S$ and vertices~$p,q\in N_G(C)\subseteq S$ as well as a~$p-q$ path~$P$ such that either
\begin{itemize}
\item~$P$ has an odd number of internal vertices and~$c(p)\neq c(q)$, or
\item~$P$ has an even number of internal vertices and~$c(p)=c(q)$.
\end{itemize}
Furthermore, all internal vertices of~$P$ are from~$V(G) \setminus S$ and~$P$ is simple except possibly for~$p=q$ (in the latter case~$P$ is in fact an odd cycle through~$p$).
\end{lemma}


Now, for instructive purposes, consider an instance~$(G,X,M,\ell)$ of the annotated problem and assume that there is a connected component~$C$ of~$G-X$ such that the parity of all paths between vertices of~$X$ which run through~$C$ matches annotations: e.g., if there is an odd~$p-q$ path,~$p,q\in X$, with internal vertices from~$C$ then~$p$ and~$q$ are annotated as monochromatic,~$\{p,q\}\in M$ (resp.\ for an even path we already have~$\{p,q\}\in E(G)$). Since~$C$ is bipartite, Lemma~\ref{lemma:twocloringextension} now implies that any~$2$-coloring of~$G[X]$ that respects all annotations can be extended to a proper~$2$-coloring of~$G[X\cup V(C)]$, i.e., extended onto~$C$.

Thus, since the components of~$G-X$ are already bipartite, we are only interested in paths between vertices of~$X$ that they provide, in particular in paths that do not match annotations. The following definition formalizes these as \emph{$X$-paths} and \emph{important~$X$-paths}.


\begin{definition} \label{importantXPath}
An \emph{$X$-path} of length~$r$ between (not necessarily distinct) vertices~$p, q \in X$ in an instance of the annotated problem is a simple path~$P = \{v_1, \ldots, v_r\}$ in~$G - X$ such that there are distinct edges~$\{p, v_1\}, \{v_r, q\} \in E(G)$. A~$p-q$~$X$-path is \emph{important} if (a) its length is odd,~$p \neq q$, and~$\{p,q\} \not \in M$, or (b) its length is even and~$\{p,q\} \not \in E(G)$.
\end{definition}
Observe that the definition of a~$p-q$~$X$-path excludes the possibility where~$p = q$ and the odd~$p-p$~$X$-path~$P$ consists of only one vertex~$v_1 = v_r$, because in that case the edges~$\{p,v_1\}$ and~$\{v_r, q\}$ are not distinct.

With the following lemma we begin to explore the structure of the important~$X$-paths. Given a graph~$G$ and a set~$X$ such that~$G-X$ is bipartite we count vertex-disjoint odd and even length~$p-q$~$X$-paths for all~$p,q\in X$. For each pair and parity the lemma will provide in polynomial time either a small hitting set intersecting all important~$X$-paths, or point out that the number of paths exceeds our budget of~$\ell$ vertex deletions (this is indicated by the sets~$A$,~$B$, and~$C$ which will later be turned into annotations, edges, and vertex deletions); Algorithm~\ref{algorithm:computehittingset} shows this in detail. We remark that both the lemma and the algorithm can also be applied to any other parameterization of \oct, given that~$X$ is a deletion set to any class of bipartite graphs (it is easy to see that~$\ell<|X|$ in all interesting cases); in particular it can be applied to the standard parameterization whose kernelizability is still open.

\begin{lemma}\label{hittingsetandannotations}
Let~$G$ be a graph,~$\ell$ be an integer, and~$X\subseteq V(G)$ such that~$G-X$ is bipartite. Then ComputeHittingSet$(G,X,\ell)$ computes sets~$A,B\subseteq\binom{X}{2}$, a set~$C\subseteq X$, and a set~$H \subseteq V(G) \setminus X$ of size at most~$4\ell\cdot|X|^2$ such that for all~$\{u,v\}\in\binom{X}{2}$:
\begin{enumerate}
\item If~$\{u,v\}\in A$ (resp.~$\{u,v\}\in B$) then there are at least~$\ell+1$ vertex-disjoint~$X$-paths of even (odd) length between~$u$ and~$v$.
\item The set~$H$ intersects all even (odd) length~$u-v$~$X$-paths with~$\{u,v\}\notin A$ (resp.~$\{u,v\}\notin B$).
\end{enumerate}
Furthermore, if~$v\in C$ then there are at least~$\ell+1$ even~$v-v$~$X$-paths (i.e., odd cycles that intersect only in~$v$), and~$H$ intersects all such paths for~$v\in X\setminus C$.  
\end{lemma}


\begin{algorithm}[t]
\caption{ComputeHittingSet$(G,X,\ell)$\label{algorithm:computehittingset}}
\begin{algorithmic}
\REQUIRE{A graph~$G$ and vertex subset~$X \subseteq V(G)$ such that~$G - X$ is bipartite.}
\ENSURE{Three sets of annotations~$A$,~$B$, and~$C$ as well as a hitting set~$H$.}

\vspace{0.2cm}
\STATE Initialize~$H,A,B,C:=\emptyset$
\STATE Let~$P \cup Q$ be a bipartition of~$G - X$ \COMMENT{Computable by BFS}
\FOREACH{$\{u,v\} \in \binom{X}{2}$}
	\STATE $PP :=$ VertexCut$(G,P,Q;u,v,P,P)$
	\STATE $QQ :=$ VertexCut$(G,P,Q;u,v,Q,Q)$
	\STATE $PQ :=$ VertexCut$(G,P,Q;u,v,P,Q)$
	\STATE $QP :=$ VertexCut$(G,P,Q;u,v,Q,P)$
	\IF[$> \ell$ disjoint even-length~$u-v$~$X$-paths]{$|PQ| > \ell$ or~$|QP| > \ell$}
		\STATE $A := A \cup \{\{u, v\}\}$
	\ELSE[Set~$PQ \cup QP$ intersects all even-length~$u-v$~$X$-paths]
		\STATE $H := H \cup (PQ \cup QP)$
	\ENDIF
	\IF{$|PP| > \ell$ or~$|QQ| > \ell$}
		\STATE $B := B \cup \{\{u, v\}\}$
	\ELSE[Set~$PP \cup QQ$ intersects all odd-length~$u-v$~$X$-paths]
		\STATE $H := H \cup (PP \cup QQ)$
	\ENDIF
\ENDFOR
\FOREACH{$v\in X$}
	\STATE $PQ :=$ VertexCut$(G,P,Q;v,v,P,Q)$
	\IF[$> \ell$ disjoint even-length~$v-v$~$X$-paths]{$|PQ| > \ell$}
		\STATE $C := C \cup \{v\}$
	\ELSE[Set~$PQ$ intersects all even-length~$v-v$~$X$-paths]
		\STATE $H := H \cup PQ$ 
	\ENDIF
\ENDFOR
\RETURN $(A, B, C, H)$
\end{algorithmic}
\end{algorithm}

\begin{algorithm}[t]
\caption{VertexCut$(G,P,Q;u,v,S,T)$}
\begin{algorithmic}
\REQUIRE{A graph~$G$ such that~$G[P\cup Q]$ is bipartite with bipartition~$P\cup Q$, vertices~$u,v\in V(G)\setminus(P\cup Q)$, and sets~$S,T\in\{P,Q\}$.}
\ENSURE{A cut~$Y\subseteq P\cup Q$ separating~$N_G(u) \cap S$ from~$N_G(v) \cap T$ in~$G-Y$.}
\vspace{0.2cm}
	\STATE Let~$G':=G[P\cup Q]$
	\STATE Add a source~$s$ with~$N_{G'}(s) := N_G(u) \cap S$ and a sink~$t$ with~$N_{G'}(t) := N_G(v) \cap T$
	\STATE Compute a minimum-size~$s-t$ vertex-cut~$Y$ in~$G'$ using a flow algorithm
\RETURN $Y$
\end{algorithmic}
\end{algorithm}

Now, let us see how to turn the sets~$A$,~$B$, and~$C$ into annotations, edges, and vertex deletions such that~$H$ is a hitting set for all important~$X$-paths in the resulting annotated instance, i.e.,~$H$ will intersect each important~$X$-path.

\begin{lemma}\label{lemma:makeannotations}
Let~$(G,X,\ell)$ be an instance of \biptwoct and let~$A,B\subseteq \binom{X}{2}$, let~$C\subseteq X$, and let~$H\subseteq V(G) \setminus X$ as given by Lemma~\ref{hittingsetandannotations}. Then in polynomial time one can find an equivalent instance~$(G', X', M, \ell')$ with $X' \subseteq X$ and $\ell' \leq \ell$ of the annotated problem such that~$H$ intersects all important~$X'$-paths in~$G'$.
\end{lemma}


We will now turn our attention to the relation between the connected components of~$G-X$ and the set~$H$ intersecting all important~$X$-paths. It is obvious that no component of~$(G-X)-H$ contains an important~$X$-path. However, to use the fact that each such path needs to leave the component via a vertex of~$H$ and cross at least one other component before returning to~$X$, we need to restrict the number of neighbors that any such component has in~$H$. This is also the point from which on we need to use that~$G-X$ has bounded treewidth. The following lemma, following along the lines of the protrusion partitioning lemma~\cite[Lemma 2]{BodlaenderFLPST09} of Bodlaender et al., permits us to extend the set~$H$ slightly while decreasing the neighborhood size of the components obtained.

\begin{lemma} \label{protrusionDecomposition} 
Let~$G$ be a graph, let~$\T$ be a tree decomposition of~$G$ of width~$w$, and let~$S\subseteq V(G)$. There is a polynomial-time algorithm that, given~$(G,\T,S)$, computes a superset~$S'\supseteq S$ of size at most~$2(w+1)|S|$ such that for each connected component~$C$ of~$G-S'$ it holds that~$|N_G(C) \cap S'| \leq 2w$.
%
\end{lemma}


The following lemma bounds the number of components of~$(G - X) - H$, regardless of the structure of the set~$H$; similar but simpler than Lemma~\ref{hittingsetandannotations}.

\begin{lemma}\label{lemma:numberofcomponents}
Let~$(G,X,M,\ell)$ be an instance of the annotated problem and let~$H$ be a set of vertices of~$G$. By deleting connected components of~$(G-X)-H$ one can in polynomial time create an equivalent instance~$(G',X,M,\ell)$ such that~$(G'-X)-H$ has at most~$2\cdot(\ell+1)\cdot(|X|+|H|)^2$ connected components.
\end{lemma}


With the next lemma, we prepare the ground for applying the combinatorial bounds on the number of cut characteristics in labeled graphs. It formalizes and proves the fact that we may freely modify any given odd cycle transversal by replacing its intersection with a connected component with a separator of the same cut characteristic. It is crucial that all important paths must intersect the hitting set~$H$ and that each component is adjacent to only few vertices of~$H$; the hitting set will correspond to terminals of certain labeled graphs, whose labels express adjacency to~$X$.

\begin{lemma}[Separator replacement lemma]\label{lemma:separatorreplacement}
Let~$(G,X,M,\ell)$ be an instance of the annotated problem. Let~$H \subseteq V(G) \setminus X$ be a set of vertices that intersects all important~$X$-paths of the instance. Let~$R$ be a solution to the problem, i.e., an odd cycle transversal such that~$G-R$ has a proper~$2$-coloring respecting the annotations. Consider a connected component~$C$ of the graph~$(G - X) - H$ and consider the terminal vertices~$N_G(C) \setminus X$. Define~$D$ as the subgraph of~$G$ induced by the set~$N_G[C] \setminus X$. Let~$T = t_1, \ldots, t_n$ be a sequence containing the terminals~$N_G(C) \setminus X$ in an arbitrary order. We define a labeling for the graph~$D$ as follows. The set of labels is the set of vertices in the modulator~$X$ augmented with one label per terminal in~$T$, and the labeling function~$f$ is defined as follows for~$v \in V(D)$:
\begin{equation*}
f(v) := 
\begin{cases}
N_G(v) \cap X & \mbox{If~$v \not \in T$.} \\
(N_G(v) \cap X) \cup \{v\} & \mbox{If~$v \in T$.}
\end{cases}
\end{equation*}
Let~$S := V(C) \cap R$ be the vertices from~$C$ chosen in the solution~$R$. If~$S' \subseteq V(C)$ is a subset such that~$S$ and~$S'$ have the same cut characteristic in the labeled graph~$(D, X \cup T, f)$ with respect to the terminals~$T$, then~$R' := (R \setminus S) \cup S'$ is also a valid solution, or more formally: if~$\K(S, T) = \K(S', T)$ with respect to the labeled graph~$(D, X \cup T, f)$ then~$G - R'$ has a proper~$2$-coloring respecting the annotations.
\end{lemma}


Now, we will use the Separator Replacement Lemma and the combinatorial bound on the number of cut characteristics (\thmref{simpleCutCharacteristicBound}) to limit the choice of vertices that may be deleted from the connected components. The idea is that it suffices to have one separator for each cut characteristic; vertices outside these separators need not be considered for deletion. To this end we introduce a restricted version of the annotated odd cycle transversal problem. As an additional restriction a set~$Z$ of vertices is provided, and the task is to find a (small) odd cycle transversal that is a subset of~$Z$.

\introduceparameterizedproblem
{Restricted Annotated (BIP~$\boldsymbol{\cap}$~\rm{$\boldsymbol{\gtw{w}}$})\bf{-OCT}}
{A graph~$G$, a set~$X\subseteq V(G)$ such that~$G-X \in (\bipartite \cap \gtw{w})$, a set~$Z \subseteq V(G)$ of deletable vertices, a set~$M\subseteq\binom{X}{2}$, and an integer~$\ell$.}
{$k:=|X|$.}
{Is there a set~$S \subseteq Z$ of size at most~$\ell$ such that~$G - S$ is bipartite, and there is a proper~$2$-coloring~$c$ of~$G - S$ such that~$c(p)=c(q)$ for all~$\{p,q\}\in M$?}

\begin{lemma}\label{lemma:restrictedproblem}
Let~$(G,X,M,\ell)$ be an instance of \annotatedbiptwoct and let~$H\subseteq V(G)\setminus X$ be a set of vertices such that:
\begin{enumerate}
\item~$H$ intersects all important~$X$-paths of~$G$,
\item~$(G-X)-H$ has at most~$\alpha$ connected components,
\item and each connected component of~$(G-X)-H$ has at most~$\delta$ neighbors in~$H$.
\end{enumerate}
For each fixed value of~$\delta$ it is possible to compute in polynomial time an equivalent instance~$(G,X,M,\ell,Z)$ of \restrictedbiptwoct where~$|Z| \leq |X| + |H| + \alpha \cdot \delta \cdot \kappa(\delta, \delta - 1, |X| + \delta)$, with~$\kappa$ as defined in \thmref{simpleCutCharacteristicBound}.
\end{lemma}


This final lemma provides the reduction from the restricted annotated problem back to \biptwoct. The number of vertices in~$Z$ in the restricted instance determines the size of the vertex set in the new (equivalent) instance.

\begin{lemma}\label{lemma:backtransformation}
An instance~$(G,X,M,\ell,Z)$ of \restrictedbiptwoct can be transformed in polynomial time into an equivalent instance $(G',X',\ell)$ of \biptwoct with~$|V(G')|$ bounded by~$|Z|+(\ell+1)\cdot|Z|^2$.
\end{lemma}


Now we can wrap up our kernelization with the following theorem. The kernelization follows the lemmas and motivation given so far.

\begin{theorem} \label{theorem:polyKernelProof}
For each fixed integer~$w \geq 1$ the problem \biptwoct admits a polynomial kernel with~$\BigO(k^{\BigO(w^3)})$ vertices.
\end{theorem}


\subsubsection{Approximating a minimum-size deletion set.} 
For our kernelization we have assumed that a deletion set~$X$ to the class~$(\bipartite \cap \gtw{w})$ is given. If~$G$ is a graph for which the minimum size of such a deletion set is~$\opt$, then we can compute in polynomial time a deletion set of size~$\BigO(\opt \cdot \log^{3/2} \opt)$ as follows. Observe that~$\gtw{w}$ is characterized by a finite set of forbidden minors, and excludes at least one planar graph as a minor. We can use the recent approximation algorithm by Fomin et al.~\cite{FominLMPS11} to approximate a deletion set~$S_{\tw(w)}$ to a graph of treewidth at most~$w$. Then we may find a minimum-size odd cycle transversal~$S_{\oct}$ in the bounded-treewidth graph~$G - S_{\tw(w)}$ which can be computed in polynomial time using Courcelle's theorem, since~$w$ is a constant. The union~$X := S_{\tw(w)} \cup S_{\oct}$ is then a suitable deletion set, which we can use to run our kernelization. This procedure is formalized in the following lemma.

\begin{lemma}
Let~$w \geq 1$ be a fixed integer. There is a polynomial-time algorithm which gets as input a graph~$G$, and computes a set~$X \subseteq V(G)$ such that~$G - X \in \bipartite \cap \gtw{w}$  with~$|X| \in \BigO(\opt \cdot \log^{3/2} \opt)$, where~$\opt$ is the minimum size of such a deletion set.
\end{lemma}


\section{Lower bounds for kernelization}\label{section:lowerbounds}

In this section we state the lower bound results for various structural kernelizations of \oct. All results use the recent notion of cross-composition introduced by Bodlaender et al.~\cite{BodlaenderJK11}. It extends the notion of a composition, showing that a reduction of the OR of any NP-hard problem into an instance of the target parameterized problem with \emph{small} parameter value excludes polynomial kernels, assuming that~\ncontainment.


\begin{theorem}\label{theorem:lowerbounds}
Assuming \ncontainment the following parameterized problems do not admit polynomial kernels:
\begin{itemize}
	\item~$(\outerplanar)$-\oct (\thmref{theorem:lowerbound:outerplanar} in the appendix),
	\item~$(\cluster)$-\oct (\thmref{theorem:lowerbound:cluster} in the appendix),
	\item~$(\cocluster)$-\oct (\thmref{theorem:lowerbound:cocluster} in the appendix),
	\item \textsc{Weighted Odd Cycle Transversal parameterized by the size of a vertex cover} (\thmref{theorem:lowerbound:weightedoctbyvertexcover} in the appendix).
\end{itemize}
\end{theorem}





\defaultbibliographystyle{abbrv}
\defaultbibliography{../Paper}

\section{Conclusion}\label{section:conclusion}
We have studied the existence of polynomial kernels for structural parameterizations of \oct. We have shown that in polynomial time the size of an instance~$(G, \ell)$ of \oct can be reduced to a polynomial in the minimum number of vertex deletions needed to transform~$G$ into a bipartite graph of constant treewidth. We also gave several kernelization lower bounds when the parameter measures the vertex-deletion distance to a non-bipartite graph with a simple structure. These lower bounds show that even for very large parameters such as the deletion distance to a cluster graph, it is unlikely that \oct admits a polynomial kernel.

The important open problem remains to determine whether the natural parameterization \loct admits a deterministic polynomial kernel. Encouraged by the recent randomized kernelization result~\cite{KratschW11}, we believe this to be the case. We think that several components we introduced in this work, such as the notion of important $X$-paths and the algorithm to find a small hitting set for these paths, will be useful ingredients for a deterministic kernelization. These ingredients do not rely on our structural parameterization and are therefore directly applicable to the general \loct problem.

\putbib[../Paper]
\end{bibunit}




\newpage
\appendix

\begin{bibunit}[alpha]

\section{On the number of cut characteristics in a labeled graph}

\newcommand{\combinatorialproperties}{
For ease of reading we present all the material of this section in the original form, introducing concepts when they are needed in the natural flow of the material. This involves a repetition of some material from the main text.

\begin{definition} \label{disconnectingSet}
Let~$G$ be a graph and let~$s,t \in V(G)$ be distinct non-adjacent vertices. A set~$S \subseteq V(G)$ is an~$s-t$ \emph{vertex-cut} if~$s,t \not \in S$ and~$S$ intersects each~$s-t$ path in~$G$.
\end{definition}

\begin{theorem}[Menger's Theorem \cite{Schrijver03}, Corollary 9.1a] \label{mengersTheorem}
If~$G$ is a graph and~$s,t \in V(G)$ are distinct non-adjacent vertices then the maximum number of internally vertex-disjoint~$s-t$ paths is equal to the minimum size of an~$s-t$ vertex-cut.
\end{theorem}

\begin{definition}
A \emph{labeled graph} is a tuple~$(G, L, f)$ where~$G$ is a graph,~$L$ is a finite set of labels, and~$f \colon V(G) \to 2^L$ is a labeling function which assigns to each vertex a (possibly empty) subset of the labels. For a subset of labels~$J \subseteq L$ we use the abbreviation~$V^f_G(J) := \{ v \in V(G) \mid f(v) \cap J \neq \emptyset\}$ to denote the vertices which carry a label from~$J$.
\end{definition}
For readability we omit the superscript on the term~$V^f_G(J)$ when this does not lead to confusion. We also use the concept of an important separator as introduced by Marx~\cite{Marx06c}.
\begin{definition}[Important separators] \label{importantSeparatorsDef}
Let~$G$ be a graph. For subsets~$X, S \subseteq V(G)$ the set of vertices reachable from~$X \setminus S$ in~$G - S$ is denoted by~$R_G(X, S)$. For~$X, Y \subseteq V(G)$ the set~$S$ is called an~\emph{$(X,Y)$-separator} if~$Y \cap R(X, S) = \emptyset$. An~$(X,Y)$-separator is \emph{minimal} is none of its proper subsets is an~$(X,Y)$-separator. An~$(X,Y)$-separator~$S'$ \emph{dominates} an~$(X,Y)$-separator~$S$ if~$|S'| \leq |S|$ and~$R(X, S) \subsetneq R(X, S')$ (proper subset). A subset~$S$ is an \emph{important}~$(X,Y)$-separator if it is minimal, and there is no~$(X,Y)$-separator~$S'$ that dominates~$S$.
\end{definition}
Carefully observe the boundary cases of this definition; note in particular that for every~$X \subseteq V(G)$ the empty set~$\emptyset$ is an important~$(X, \emptyset$)-separator. To improve readability we will write~$(t, Y)$-separator instead of~$(\{t\}, Y)$-separator when the first set is just a singleton. Similarly we will write~$R_G(t, S)$ instead of~$R_G(\{t\}, S)$.

We use several results from recent work by Marx and Razgon~\cite{MarxR10}. Although the notation used in the recent work is slightly different than in the original paper by Marx~\cite{Marx06c}, the results also hold for the original notation that we use here. The following claim originates from~\cite[Proposition 2.5]{MarxR10}, and follows from the given definitions.
\begin{proposition} \label{importantInInducedGraphs}
Let~$G$ be a graph, let~$X,Y \subseteq V(G)$ and let~$S$ be an important~$(X,Y)$-separator. For all~$v \in S$ it holds that~$S \setminus \{v\}$ is an important~$(X \setminus \{v\}, Y \setminus \{v\})$-separator in the graph~$G - \{v\}$.
\end{proposition}

In the original paper on important separators, Marx showed a bound of~$4^{k^2}$ on the number of important~$(X,Y)$-separators of size at most~$k$. An improvement of this bound was implicit in work by Chen et al.~\cite{ChenLL09}, which was summarized and made explicit by Marx and Razgon~\cite[Lemma 2.6]{MarxR10}. This bound also holds for the original definition of important separators.
\begin{lemma}[\cite{MarxR10}] \label{numberImportantSeparators}
If~$X$ and~$Y$ are arbitrary vertex sets of a graph~$G$ then there are at most~$4^m$ important~$(X,Y)$-separators of size at most~$m \geq 0$.
\end{lemma}

 We need the following simple lemma about minimal separators.
\begin{lemma} \label{minimalSeparatorsPaths}
Let~$G$ be a graph, and let~$X, Y \subseteq V(G)$ be vertex subsets. If~$S \subseteq V(G)$ is a \emph{minimal}~$(X,Y)$-separator then for every~$s \in S$ it must hold that~$s \in R(X, S \setminus \{s\})$.
\end{lemma}
\begin{proof}
If~$S$ is an~$(X,Y)$-separator containing~$s \in S$ with~$s \not \in R(X, S \setminus \{s\})$ then~$S \setminus \{s\}$ is also an~$(X,Y)$-separator, showing that~$S$ is not minimal.
\qed
\end{proof}


We also need a lemma on the structure of important separators which intersect the set they are separating.
\begin{lemma} \label{importanceForSubsets}
Let~$G$ be a graph, and let~$X, Y \subseteq V(G)$ be vertex subsets. If~$S \subseteq Y$ is an important~$(X,Y)$-separator then for every set~$Y'$ satisfying~$S \subseteq Y' \subseteq Y$ the set~$S$ is an important~$(X,Y')$-separator.
\end{lemma}

\begin{proof}
Since~$Y' \subseteq Y$ it is easy to see that~$S$ separates~$X$ from~$Y'$. Since~$S$ is an important~$(X,Y)$-separator, it is a minimal $(X,Y)$-separator which implies by \lemmaref{minimalSeparatorsPaths} that for all~$s \in S$ we have~$s \in R(X, S \setminus \{s\})$. Since~$S \subseteq Y'$ this proves that~$S$ is a minimal $(X,Y')$-separator. If there is an $(X,Y')$-separator~$S'$ which dominates~$S$, then~$R(X, S)$ must be a proper subset of~$R(X,S')$; but this is only possible if $S \cap R(X,S') \neq \emptyset$ which implies by~$S \subseteq Y'$ that~$S'$ is not, in fact, an $(X,Y')$-separator. This concludes the proof.
\qed
\end{proof}

\begin{lemma} \label{pathsForImportantSeparator}
Let~$G$ be a graph, let~$X,Y\subseteq V(G)$, and let~$S \subseteq V(G)$ be an important~$(X,Y)$-separator with~$S \nsubseteq Y$. Then there is a set~$\P = \{P_0, P_1, \ldots, P_{|S|}\}$ of distinct simple paths such that:
\begin{enumerate}
	\item each path~$P_i$ connects a vertex of~$Y$ and a vertex of~$S$,
	\item the paths~$P_1, \ldots, P_{|S|}$ are pairwise vertex-disjoint, i.e.\ $V(P_i) \cap V(P_j) = \emptyset$ for~$1 \leq i < j \leq |S|$, and
	\item the path~$P_0$ is vertex-disjoint from the paths~$P_2, \ldots, P_{|S|}$, i.e.\ $V(P_0) \cap V(P_i) = \emptyset$ for~$2 \leq i \leq |S|$,
	\item the paths~$P_0$ and~$P_1$ intersect only in their endpoint in the set~$S$, i.e.\ $V(P_0) \cap V(P_1) = \{s\}$ with~$s \in S$.
\end{enumerate}
\end{lemma}

\begin{proof}
Let~$S_0:=S\cap Y$ and let~$S_1:=S\setminus Y$; clearly~$S_1\neq \emptyset$.
Let a graph~$G'$ be obtained from~$G$ in the following way. Delete the set~$S_0$ and then add a vertex~$s^*$ with neighborhood~$N_{G'}(s^*) := N_G(S_1) \setminus S_0$. Let~$S' := S_1 \cup \{s^*\}$. Add a source~$p$ with~$N_{G'}(p) := Y \setminus S$ and add a sink~$q$ with~$N_{G'}(q) := S'$. (Note that~$G-S_0$ is a subgraph of~$G'$, i.e.~$G-S_0=G'-\{p,q,s^*\}$.)
Let~$\P'$ be a maximum packing of internally vertex-disjoint~$p-q$ paths in~$G'$. By Menger's theorem (\thmref{mengersTheorem}) the size~$|\P'|$ equals the size of a minimum~$p-q$ vertex-cut. Observe that there are at most~$|S'|=|S_1|+1$ such paths, matching the degree of~$q$. The main part of the proof is devoted to showing that~$|\P'| = |S'| = |S_1| + 1$ when~$S$ is an important~$(X,Y)$-separator. We will prove this claim, and afterwards we show how to obtain a set of paths in~$G$ as mentioned in the statement of the lemma. 

So assume for a contradiction that~$|\P'|\leq |S_1|$ and let~$A\subseteq V(G')\setminus\{p,q\}$ be a corresponding minimum~$p-q$ vertex-cut in~$G'$, which has size equal to~$|\P'| \leq |S_1|$ by Menger's theorem. Observe that~$A \neq S_1$: if there is a~$p-q$ path in~$G'$, then there is such a path which avoids~$S_1$ since any~$p-q$ path that reaches~$q$ via~$S_1\subseteq N_{G'}(q)$ can be routed via~$s^*$, avoiding~$S_1$. Indeed, such a path must be of the form~$(p,y,\ldots,s,q)$ with~$y\in Y$,~$s\in S_1$, and with~$y\neq s$ since~$Y\cap S_1=\emptyset$. Hence it reaches~$S_1$ from a vertex in~$N_G(S_1) \setminus S_0$ implying that there is also a path~$(p,y,\ldots,s^*,q)$ avoiding~$A=S_1$. This shows that if there is a~$p-q$ path in~$G'$ then~$S_1$ is not a~$p-q$ vertex-cut in~$G'$ and therefore~$A \neq S_1$; if there is no~$p-q$ path in~$G'$ then the non-empty set~$S_1$ cannot be a minimum-size~$p-q$ vertex-cut (the empty set is a minimum-size vertex-cut), and so we again find~$A \neq S_1$.

Knowing that~$A\neq S_1$, we will now show that (i)~$A\cup S_0$ is an~$(X,Y)$-separator in~$G$ and (ii) that the~$(X,Y)$-separator~$A\cup S_0$ dominates the separator~$S$.

\textbf{i)} Let~$P$ be any~$x-y$ path in~$G$ for some~$x\in X$ and~$y\in Y$. If~$P$ contains a vertex of~$S_0\subseteq A\cup S_0$ we are done. Otherwise there must be a vertex~$w\in S_1=S\setminus S_0$ on~$P$, i.e.,~$P=(x,\ldots,w,\ldots,y)$, since~$S$ separates~$X$ and~$Y$. Furthermore,~$P$ must exist also in~$G'$ (which contains~$G-S_0$). Hence, there must be a vertex of~$A$ on~$P$ separating~$w$ and~$y$, since~$A$ separates~$S_1$ from~$Y$ in~$G'$. Thus~$A\cup S_0$ separates~$X$ and~$Y$ in~$G$.

\textbf{ii)} We now prove that~$A \cup S_0$ dominates~$S$. Assume for a contradiction that~$A$ contains a vertex~$v \in R(X, S)$; since it follows from \defref{importantSeparatorsDef} that~$S \cap R(X, S) = \emptyset$ we must have~$v \not \in S$. Since~$A$ is a minimum~$p-q$ vertex-cut in~$G'$ there must be a~$p-q$ path~$P$ in~$G'$ that intersects~$A$ only in~$v$. Furthermore,~$P$ must leave~$p$ to some vertex~$y\in Y \setminus S = N_{G'}(p)$, i.e.,~$P=(p,y,\ldots,v,\ldots,q)$. Note that~$y\neq v$, otherwise~$X \setminus S$ would reach~$Y \setminus S$ in~$G-S$.

Since~$v \in R(X,S)$ but~$S$ is an~$(X,Y)$-separator, there must be a vertex~$w$ of~$S$ between~$y$ and~$v$ on~$P$, i.e.,~$P=(p,y,\ldots,w,\ldots,v,\ldots,q)$. Since~$w$ is a vertex of~$G'$ it follows that~$w\notin S_0=S\cap Y$. Thus~$w\in S_1$,~$w\neq y\in Y$, and~$w\neq v\notin S$. However,~$w$ is adjacent to~$q$ in~$G'$ which gives rise to a~$p-q$ path~$(p,y,\ldots,w,q)$ in~$G'$ that is not intersecting~$A$ since~$v$ is not contained in it (and~$P$ contains no other vertices of~$A$), a contradiction. Thus,~$A$ contains no vertices reachable from~$X$ in~$G-S$ and the same must be true for~$A\cup S_0$ as~$S_0\subseteq S$. Therefore we have~$R(X, S) \subseteq R(X, A \cup S_0)$.

Now, to see that~$R(X,S)$ is a proper subset of~$R(X, A \cup S_0)$ let us fix a vertex~$v\in S_1\setminus A$ (such a vertex exists since~$|A| \leq |S_1|$ and~$A \neq S_1$). If~$v \in X$ then~$v \in R(X, A \cup S_0) \setminus R(X, S)$ which proves the proper subset relation. If~$v \not \in X$ then since~$S$ is minimal, there must be a vertex~$u$ which is adjacent to~$v$ in the set~$R(X, S)$ (otherwise~$S\setminus\{v\}$ would also be an~$(X,Y)$-separator). Therefore we must have~$u \in R(X, A \cup S_0)$ (as shown above) and since~$v \not \in A \cup S_0$ this shows that~$v \in R(X, A \cup S_0)$, completing the argument that~$A\cup S_0$ dominates~$S$.

We have seen that the assumption that~$|\P'| \leq |S_1|$ leads to the construction of an~$(X,Y)$-separator separator~$A \cup S_0$ which dominates~$S$, thereby proving that~$S$ is not an important separator. Hence if~$S$ is an important~$(X,Y)$-separator we must find a packing~$\P'$ of~$|S_1|+1$ internally vertex-disjoint~$p-q$ paths in graph~$G'$, as claimed.

We will now define a packing~$\P$ of paths in~$G$ that matches our claim. Let~$\P''$ be obtained from~$\P'$ by removing the vertices~$p$ and~$q$ from each path; it follows that~$\P''$ is a packing of vertex-disjoint paths which connect~$N_{G'}(p) = Y \setminus S$ to~$N_{G'}(q) = S_1 \cup \{s^*\}$ in graph~$G'$; since~$p$ and~$q$ are not adjacent in~$G'$ all paths in~$\P''$ are non-empty, and since~$N_{G'}(p) \cap N_{G'}(q) = \emptyset$ all such paths have at least two vertices. Since~$|\P''| = |S_1| + 1$ we know that each path in~$\P''$ ends in a unique vertex of~$S_1 \cup \{s^*\}$. Denote by~$P'_0$ the path which ends in vertex~$s^*$. Since~$N_{G'}(s^*) = N_G(S_1) \setminus S_0$ there is at least one vertex~$s' \in S_1$ which is adjacent in~$G$ to the predecessor of~$s^*$ on path~$P'_0$. Let~$P_1$ be the path in~$\P''$ which ends in~$s'$, and let~$P_2, \ldots, P_{|S_1|}$ be the remaining paths of~$\P''$ in arbitrary order. Each of the paths~$P_1, \ldots, P_{|S_1|}$ also exists in graph~$G$. If we replace the occurence of vertex~$s^*$ on path~$P'_0$ by vertex~$s'$ to obtain path~$P_0$, then the resulting path~$P_0$ also exists in graph~$G$ and only intersects~$P_1$ in the single vertex~$s' \in S$. We are now ready to define the final packing of paths~$\P$. We start with the paths~$P_0, P_1, \ldots, P_{|S_1|}$. For each vertex~$v \in S_0$ we add a singleton path on vertex~$v$ - since~$v \in S_0 = Y \cap S$ such a path trivially connects a vertex from~$Y$ to a vertex in~$S$, and because the vertices of~$S_0$ do not exist in the graph~$G'$ from which the other paths were taken, the new paths we add in this way are vertex-disjoint from the others. The resulting set of paths has size~$1 + |S_1| + |S_0| = |S| + 1$, and it is easy to verify that these paths satisfy the stated claims on disjointness.
\qed
\end{proof}

\begin{lemma} \label{smallWitnessImportantSeparator}
Let~$(G, L, f)$ be a labeled graph, let~$X$ be a subset of vertices and let~$J \subseteq L$ be a subset of labels. If~$S$ is an important~$(X, V_G(J))$-separator of size~$k$ then there is a set~$J^* \subseteq J$ with~$|J^*| \leq \sum _{i=1}^k (i+1)$ such that~$S$ is an important~$(X, V_G(J^*))$-separator.
\end{lemma}

\begin{proof}
We prove the statement by induction on~$k$.

\paragraph{Base case} If~$k = 0$ then~$S$ is the empty set, which happens when the vertices carrying labels from~$J$ do not occur in the same connected component as vertices of~$X$. Take~$J^* := \emptyset$ which implies that~$V_G(J^*) = \emptyset$. It follows from the definition of important separators that~$S = \emptyset$ is an important~$(X, \emptyset)$-separator, which proves the base case.

\paragraph{Induction step} Consider the more interesting case that~$k > 0$, and let~$S = \{s_1, \ldots, s_k\}$.

We first handle the case that~$S \subseteq V_G(J)$. So assume that~$S \subseteq V_G(J)$, which implies that each vertex from~$S$ carries at least one label from~$J$. Let~$z_i \in J \cap f(s_i)$ be a label from~$J$ carried by vertex~$s_i$ for~$1 \leq i \leq k$, and define~$J^* := \{z_1, \ldots, z_k\}$. It is clear that~$J^* \subseteq J$ and~$|J^*| \leq k$. It is not hard to see that~$S \subseteq V_G(J^*) \subseteq V_G(J)$. \lemmaref{importanceForSubsets} now shows that since~$S$ is an important~$(X, V_G(J))$-separator with~$S \subseteq V_G(J)$ and~$S \subseteq V_G(J^*) \subseteq V_G(J)$, the set~$S$ must also be an important~$(X, V_G(J^*))$-separator. This concludes the proof of the lemma for the case that~$S \subseteq V_G(J)$.

In the remainder we attack the harder case that~$S \not \subseteq V_G(J)$, by considering the structure of the separator~$S$. Since~$S$ is an important~$(X, V_G(J))$-separator it follows from \lemmaref{pathsForImportantSeparator} (re-numbering the vertices from~$S$ if need be) together with the assumption that~$S \not \subseteq V_G(J)$ that there is a set of~$S-V_G(J)$ paths~$P_1, \ldots, P_k$ and an additional path~$P_0$ such that:
\begin{enumerate}
	\item path~$P_i$ for~$1 \leq i \leq k$ connects a vertex~$v_i \in V_G(J)$ carrying label~$z_i \in f(v_i) \cap J$ to vertex~$s_i$ (the~$i$-th vertex of the separator~$S$),
	\item the paths~$P_1, \ldots, P_k$ are pairwise vertex-disjoint,
	\item~$V_G(P_0) \cap V_G(P_1) = \{s_1\}$ and~$V_G(P_0) \cap V_G(P_i) = \emptyset$ for~$2 \leq i \leq k$,
	\item path~$P_0$ connects vertex~$v_0 \in V_G(J)$ carrying label~$z_0 \in f(v_0) \cap J$ to vertex~$s_1$.
\end{enumerate}
The vertex~$s_1$ which is the endpoint of the two paths~$P_0$ and~$P_1$ plays a special role in our argument. By \proposref{importantInInducedGraphs} we know that the set~$S' := S \setminus \{s_1\}$ is an important~$(X \setminus \{s_1\}, V_G(J) \setminus \{s_1\})$-separator of size~$k - 1$ in the graph~$G' := G - \{s_1\}$. Since the set~$V_G(J) \setminus \{s_1\}$ contains exactly those vertices of~$G'$ carrying a label from~$J$, it follows that~$V_G(J) \setminus \{s_1\} = V_{G'}(J)$ and therefore~$S'$ is an important~$(X \setminus \{s_1\}, V_{G'}(J))$-separator in graph~$G'$. We may therefore apply induction to find a set~$J' \subseteq J$ such that~$S'$ is an important~$(X \setminus \{s_1\}, V_{G'}(J'))$-separator in graph~$G'$ with~$|J'| \leq \sum _{i=1}^{k-1} (i+1)$. We will use the set~$J'$ to build the desired set of labels~$J^*$, as follows. 

Define~$J^* := J' \cup \{z_0, z_1, \ldots, z_k\}$, from which it is easy to see that~$|J^*| \leq \sum _{i=1}^k (i+1)$. We claim that~$S$ is an important~$(X, V_G(J^*))$-separator in graph~$G$. Since~$S$ is an important~$(X, V_G(J))$-separator and~$V_G(J^*) \subseteq V_G(J)$ (since~$J^* \subseteq J$) it follows immediately that~$S$ is also an~$(X, V_G(J^*))$-separator; hence if~$S$ is not an important separator then~$S$ is not minimal or it is dominated by some other separator. So let~$S^*$ with~$|S^*| \leq k$ be an~$(X, V_G(J^*))$-separator in graph~$G$ which is either a proper subset of~$S$, or which dominates~$S$: we will derive a contradiction.

\begin{claim}
The separator~$S^*$ must contain exactly one vertex from each path~$P_i$ for~$0 \leq i \leq k$.
\end{claim}

\begin{claimproof}
We first prove by contradiction that~$S^*$ contains at least one vertex from each path. So assume that there is a path~$P_i$ such that~$V(P_i) \cap S^* = \emptyset$, and let~$s \in S$ be the endpoint of path~$P_i$ in the set~$S$. Since~$S$ is an important~$(X, V(J))$-separator in~$G$, we know by the definition of important separators that~$S$ is a \emph{minimal}~$(X, V(J))$-separator, which implies by \lemmaref{minimalSeparatorsPaths} that~$s \in R(X, S \setminus \{s\})$. As the next step we will show that~$s \in R(X, S^*)$. 
\begin{itemize}
	\item If~$s \in X$ then the assumption that~$V(P_i) \cap S^* = \emptyset$ proves together with~$s \in V(P_i)$ that~$s \not \in S^*$ and therefore~$s \in R(X, S^*)$. 
	\item If~$s \not \in X$, then since~$s \in R(X, S \setminus \{s\})$ there is a simple path~$P_{x-s}$ from a vertex~$x \in X \setminus (S \setminus \{s\}) = X \setminus S$ to the vertex~$s$ in the graph~$G - (S \setminus \{s\})$, and since~$s \not \in X$ this path must contain at least two vertices. Consider the subpath~$P_{x-u}$ of~$P_{x-s}$ which leads from~$x$ to the predecessor~$u$ of vertex~$s$ on path~$P_{x-s}$. Then~$s \not \in V(P_{x-u})$ and path~$P_{x-u}$ does not use any vertices of~$S \setminus \{s\}$ (since it is a subpath of~$P_{x-s}$ in~$G - (S \setminus \{s\})$), and therefore~$V(P_{x-u}) \cap S = \emptyset$, which together with the fact that~$x \in X \setminus S$ implies that~$V(P_{x-u}) \subseteq R(X, S)$. Since~$S^*$ is either a proper subset of~$S$, or an~$(X, V_G(J^*))$-separator which dominates~$S$, we must have~$R(X, S) \subseteq R(X, S^*)$. By combining these two facts we see that~$V(P_{x-u}) \subseteq R(X, S) \subseteq R(X, S^*)$, and since~$u$ is adjacent to~$s$ with~$s \not \in S^*$ this then proves that~$s \in R(X, S^*)$.
\end{itemize}
We now know that~$s \in R(X,S^*)$ and that path~$P_i$ starts at vertex~$s$ with~$V(P_i) \cap S^* = \emptyset$; this shows that~$V(P_i) \subseteq R(X, S^*)$. But the endpoint~$v_i~$ of path~$P_i$ carries the label~$z_i \in J^*$ and therefore~$v_i \in V_G(J^*)$ and~$v_i \in R(X, S^*)$; but this then proves that~$S^*$ is not an~$(X, V_G(J^*))$-separator, a contradiction. Hence we know that~$S^*$ contains at least one vertex from each path~$P_i$ for~$0 \leq i \leq k$.

To complete the proof we show that~$S^*$ cannot contain more than one vertex from each path~$P_i$ with~$0 \leq i \leq k$. By the structure of the paths we know that the set~$\P_1 := \{P_0, P_2, \ldots, P_k\}$ contains~$k$ paths which are mutually vertex-disjoint, and the set~$\P_2 := \{P_1, P_2, \ldots, P_k\}$ also contains~$k$ paths which are mutually vertex-disjoint. Since we already showed that~$S^*$ of size~$k$ must contain at least one vertex from each path in the set~$\P_1$ and the paths in that set are vertex-disjoint, it must contain exactly one vertex from each path in the set~$\P_1$. The same argument shows that~$S^*$ must contain one vertex from each of the~$k$ disjoint paths in~$\P_2$. But~$\P_1$ and~$\P_2$ together contain all paths~$P_i$ for~$0 \leq i \leq k$, so we have shown that~$S^*$ contains exactly one vertex on each path.
\end{claimproof}

\begin{claim}
Separator~$S^*$ must contain vertex~$s_1$.
\end{claim}

\begin{claimproof}
Recall that the paths~$P_i$ are chosen such that paths~$P_0$ and~$P_1$ only intersect at~$s_1$, and other pairs do not intersect at all. By the previous claim~$S^*$ must contain exactly one vertex from each~$P_i$. Since~$|S^*| \leq k$ and there are~$k+1$ paths, the only way this can be done is if~$S^*$ contains at least one vertex which lies on multiple paths. But~$s_1$ is the only such vertex, hence~$s_1 \in S^*$.
\end{claimproof}

Now that we have some more information about the structure of potential sets~$S^*$ we will finish the proof by showing that~$S^*$ cannot exist.

\begin{itemize}
	\item Assume first that~$S$ is not a minimal~$(X, V(J^*))$-separator because there is a proper subset~$S^* \subsetneq S$ which is also an~$(X, V(J^*))$-separator. Since~$|S| = k$ this would imply~$|S^*| < k$, but since the first claim shows that~$S^*$ must contain one vertex from each path~$P_i$ and the paths~$P_1, \ldots, P_k$ are mutually vertex-disjoint, no set of size less than~$k$ can satisfy this requirement. Hence the set~$S$ must be a \emph{minimal}~$(X, V(J^*))$-separator.
	\item For the remaining case, assume that~$S$ is a minimal but not important $(X, V(J^*))$-separator because it is dominated by a separator~$S^*$, which implies that~$R_G(X, S) \subsetneq R_G(X, S^*)$. Since~$s_1 \in S \cap S^*$ by the second claim, we know that~$R_G(X,S) = R_{G'}(X \setminus \{s_1\}, S \setminus \{s_1\})$ and similarly~$R_G(X,S^*) = R_{G'}(X \setminus \{s_1\}, S^* \setminus \{s_1\})$. Therefore~$R_{G'}(X \setminus \{s_1\},S \setminus \{s_1\}) \subsetneq R_{G'}(X \setminus \{s_1\},S^* \setminus \{s_1\})$. By \lemmaref{importantInInducedGraphs} it follows that~$S^* \setminus \{s_1\}$ is an~$(X \setminus \{s_1\}, V_{G'}(J^*))$-separator in~$G'$, and since~$J' \subseteq J^*$ this implies that~$S^* \setminus \{s_1\}$ is an~$(X \setminus \{s_1\}, V_{G'}(J'))$-separator in~$G'$ with~$R_{G'}(X \setminus \{s_1\},S \setminus \{s_1\}) \subsetneq R_{G'}(X \setminus \{s_1\},S^* \setminus \{s_1\})$ and~$|S^* \setminus \{s_1\}| \leq |S \setminus \{s_1\}| = |S'|$; but this then proves that~$S'$ is not an important~$(X \setminus \{s_1\}, V_{G'}(J'))$-separator in~$G'$ which contradicts the induction hypothesis which was invoked earlier on in the proof. Hence such a set~$S^*$ cannot dominate the separator~$S$ with respect to separation of~$(X, V(J^*))$.
\end{itemize}
We have seen that the assumption that~$S$ is not an important~$(X, V(J^*))$-separator leads to a contradiction; this concludes the proof.
\qed
\end{proof}

\begin{definition} \label{singleTerminalCutLabels}
Let~$(G, L, f)$ be a labeled graph and let~$t \in V(G)$ be a distinguished terminal vertex in~$G$. If~$S \subseteq V(G)$ is a subset of vertices then the set of \emph{labels reachable from~$t$ in~$G - S$} is defined as:
\begin{equation*}
\L(t, S) := \bigcup _{v \in R(t, S)} f(v).
\end{equation*}
\end{definition}

\begin{lemma} \label{equivalentImportantSeparatorLemma}
Let~$(G, L, f)$ be a labeled graph with a terminal vertex~$t \in V(G)$ and let~$S \subseteq V(G)$ be a subset of vertices. Let~$J := L \setminus \L(t, S)$ be the labels which are unreachable from~$t$ in~$G - S$. Then there is an important~$(t, V_G(J))$-separator~$S'$ such that~$|S'| \leq |S|$ and the sets~$S$ and~$S'$ separate the same set of labels from~$t$:~$\L(t, S) = \L(t, S')$.
\end{lemma}

\begin{proof}
Assume the conditions in the statement of the lemma hold. It follows from the definition of~$\L$ that the set~$S$ must be a~$(t, V_G(J))$-separator. Hence if~$S$ is an important~$(t, V_G(J))$-separator then taking~$S' = S$ satisfies all conditions of the lemma, and we are done. So assume for the remainder that~$S$ is not an important~$(t, V_G(J))$-separator because~$S$ is not a minimal separator, or because there is a~$(t, V_G(J))$-separator which dominates~$S$. If~$S$ is not minimal then let~$S'$ be a proper subset which is still a separator; if~$S$ is dominated then let~$S'$ a separator which dominates~$S$. In both cases it is easy to see from the definitions that~$R(t, S) \subseteq R(t, S')$ and~$|S'| \leq |S|$. This implies that for every vertex~$v \in V_G(J)$ which is reachable from~$t$ in~$G - S$, this vertex is still reachable in~$G - S'$ and therefore~$\L(t, S) \subseteq \L(t, S')$. Now observe that since~$S'$ is a~$(t, V_G(J))$-separator, by the definition of separation we know that no vertices carrying a label of~$J$ can be reachable from~$t$ after deleting~$S'$. Since the set~$J$ contains exactly those labels which were not reachable from~$t$ in~$G - S$, we know that none of these labels are reachable from~$t$ in~$G - S'$. Hence we must also have~$\L(t, S') \subseteq \L(t, S)$, which together with our earlier fact shows~$\L(t, S) = \L(t, S')$.

If~$S'$ is an \emph{important}~$(t, V_G(J))$-separator then the sets~$S'$ and~$J$ satisfy all conditions of the lemma, and we are done. If~$S'$ is not important then we can repeat the argument to find a separator~$S''$ which is a subset of~$S'$ or dominates~$S'$, and for which~$\L(t, S') = \L(t, S'')$. We can repeat this process until we have found an important separator, and since we either decrease the size of the separator or move to a dominating separator at each step, the process terminates after a finite number of steps. This proves the lemma.
\qed
\end{proof}

\begin{definition}
Let~$(G, L, f)$ be a labeled graph and let~$T = t_1, \ldots, t_n$ be a sequence of distinct terminal vertices in~$G$. The \emph{cut characteristic}~$\K(S, T)$ of a set~$S \subseteq V(G)$ with respect to the terminals~$T$ is an~$n$-dimensional vector whose elements are subsets of~$L$, and which is defined as:
\begin{equation*}
\K(S, T) := \left ( \L(t_1, S), \L(t_2, S), \ldots, \L(t_n, S) \right ).
\end{equation*}
Define the set of distinct cut characteristics~$\K^m(T)$ for separators of size at most~$m \geq 1$ as:
\begin{equation*}
\K^m(T) := \left \{ \K(S, T)  \middlemid S \in \binom{V(G)}{\leq m} \right \}.
\end{equation*}
\end{definition}

The final goal of this section is to bound~$\K^m(T)$ for arbitrary sets of terminals~$T$. As the next step we will show how to bound this term when~$T = \{t\}$ is a singleton.

\begin{lemma} \label{singleTerminalBound}
Let~$(G, L, f)$ be a labeled graph and let~$t$ be a distinguished terminal vertex in~$G$. Then~$|\K^m(\{t\})| \leq \binom{|L|}{\leq m'} 4^m$, where~$m' = \mbound$.
\end{lemma}

\begin{proof}
Assume the conditions in the statement of the lemma hold. We will define a set~$\H \subseteq \binom{V(G)}{\leq m}$ of bounded size, and show that for every~$S \in \binom{V(G)}{\leq m}$ there is a~$S' \in \H$ such that~$\L(t, S) = \L(t, S')$, which will then imply a bound on~$|\K^m(\{t\})|$. Let~$m' := \mbound$. Now define \H as follows:
\[
\H := \left \{ S \in \binom{V(G)}{\leq m} \middlemid \exists J \in \binom{L}{\leq m'}: S \mbox{ is an important~$(t, V_G(J))$-separator}  \right \}.
\]

We will show that the size of \H is bounded independently of the number of vertices in the graph~$G$. Consider some set of labels~$J' \in \binom{L}{\leq m'}$ and the vertices~$V_G(J')$ on which those labels appear. By \lemmaref{numberImportantSeparators} the number of important~$(t, V(J'))$-separators of size at most~$m$ is bounded by~$4^m$. Hence the number of separators in the set \H which are added because of this~$J'$ is at most~$4^m$. Since the number of different options for~$J$ is~$|\binom{L}{\leq m'}|$ it follows that~$|\H| \leq \binom{|L|}{\leq m'} 4^m$. It is easy to see that the set~$\{ \L(t, S) \mid S \in \H \}$ is not larger than~$\H$. To complete the proof we will therefore show that this is a superset of~$\K^m(\{t\})$. From the definition of the set~$\K^m(\{t\})$ it suffices to show that for every~$S \in \binom{V(G)}{\leq m}$ there is a set~$S' \in \H$ such that~$\L(t, S) = \L(t, S')$, which will be the subject of the remainder of the proof.

So let~$S \in \binom{V(G)}{\leq m}$. Now take~$J := L \setminus \L(t, S)$; it follows that~$S$ is a~$(t, V_G(J))$-separator. By \lemmaref{equivalentImportantSeparatorLemma} we know that there is an \emph{important}~$(t, V_G(J))$-separator $S'$ with~$|S'| \leq |S| \leq m$ and~$\L(t, S) = \L(t, S')$. By \lemmaref{smallWitnessImportantSeparator} there exists a set~$J^* \subseteq J$ satisfying~$|J^*| \leq \sum _{i=1}^{|S|} (i+1) \leq \sum _{i=1}^m (i+1) = m'$ such that~$S'$ is an important~$(t, V_G(J^*))$-separator. But since~$|J^*| \leq m'$ we must have~$J^* \in \binom{L}{\leq m'}$, and therefore~$S' \in \H$. Since~$\L(t, S) = \L(t, S')$ and~$S' \in \H$ this proves that~$\{ \L(t, S) \mid S \in \H\}$ is indeed a superset of~$\K^m(\{t\})$, and since we showed earlier that~$|\H| \leq \binom{|L|}{\leq m'} 4^m$ this concludes the proof of the lemma.
\qed
\end{proof}

\begin{lemma} \label{multiterminalSeparatorBound}
If~$(G, L, f)$ is a labeled graph and~$T = t_1, \ldots, t_n$ is a sequence of distinct terminal vertices in~$G$ then the number of distinct cut characteristics for separators of size at most~$m$ is polynomial in~$|L|$ for fixed values of~$m$ and~$n$:~$|\K^m(T)| \leq (\binom{|L|}{\leq m'} 4^m)^n$, where~$m' = \mbound$.
\end{lemma}

\begin{proof}
Assume the conditions in the lemma statement to hold. The set~$\K^m(T)$ contains~$n$-tuples of sets of labels. If we look at the set of all such~$n$-tuples and restrict our attention to column number~$i$ for~$1 \leq i \leq n$, then the elements occurring in that column are exactly the elements of the set~$\K^m(\{t_i\})$. \lemmaref{singleTerminalBound} shows that~$|\K^m(\{t_i\})| \leq \binom{|L|}{\leq m'} 4^m$. So~$\K^m(T)$ contains~$n$-tuples where the elements of the~$i$-th column are taken from a domain with at most~$\binom{|L|}{\leq m'} 4^m$ different members; this shows that the number of distinct tuples is at most~$(\binom{|L|}{\leq m'} 4^m)^n$, which concludes the proof.
\qed
\end{proof}

\thmref{simpleCutCharacteristicBound} follows directly from \lemmaref{multiterminalSeparatorBound} by simple formula manipulation.
}


\combinatorialproperties

\section{Omitted proofs of Section~\ref{section:kernelization}}

\subsection{Proof of Lemma~\ref{lemma:twocloringextension}}

\newcommand{\xpathfigure}{
\begin{figure}[t]
\centering
\begin{tikzpicture}[scale=1,thick]
\coordinate (p) at (0,3);
\coordinate (q) at (0,2);
\coordinate (r) at (0,1);
\coordinate (s) at (0,0);
\coordinate (a) at (2,3);
\coordinate (b) at (2,2);
\coordinate (c) at (2,1.5);
\coordinate (d) at (2,1);
\coordinate (e) at (2,0.5);
\coordinate (f) at (2,0);
\coordinate (g) at (4,2.5);
\coordinate (h) at (4,0.5);

\draw (q) -- (r);
\draw (p) -- (a);
\draw (q) -- (b);
\draw[dashed] (a) -- (g) -- (b);
\draw (q) -- (b) -- (c) -- (r);
\draw (r) -- (e) -- (s);
\draw (r) -- (d);
\draw[dashed] (d) -- (h) -- (f);
\draw (f) -- (s);
\draw (g) -- (h);


\draw[dotted] (0,3.5) .. controls (0.5,3.5) .. (0.5,3) -- (0.5,0) .. controls (0.5,-0.5) .. (0,-0.5) .. controls (-0.5,-0.5) .. (-0.5,0) -- (-0.5,3) .. controls (-0.5,3.5) .. (0,3.5);

\draw[dotted] (4,3.5) .. controls (4.5,3.5) .. (4.5,3) -- (4.5,0) .. controls (4.5,-0.5) .. (4,-0.5) -- (2,-0.5) .. controls (1.5,-0.5) .. (1.5,0) -- (1.5,3) .. controls (1.5,3.5) .. (2,3.5) -- cycle;

\draw (0,2.5) node {$X$};
\draw (3.8,3) node {$G-X$};

\draw[color=white!50!black] (g) +(0.2,0.2) -- +(-0.2,0.2) --  +(-0.2,-0.2) -- +(0.2,-0.2) -- cycle;
\draw[color=white!50!black] (h) +(0.2,0.2) -- +(-0.2,0.2) --  +(-0.2,-0.2) -- +(0.2,-0.2) -- cycle;

\drawgreyvertex{(p)}
\drawgreyvertex{(q)}
\drawgreyvertex{(r)}
\drawgreyvertex{(s)}
\drawemptyvertex{(a)}
\drawemptyvertex{(b)}
\drawemptyvertex{(h)}
\drawblackvertex{(c)};
\drawblackvertex{(d)};
\drawblackvertex{(e)};
\drawblackvertex{(f)};
\drawblackvertex{(g)};

\draw (g) +(-1,0) node {$P_1$};
\draw (h) +(-1,0) node {$P_2$};

\draw (p) +(-0.25,0) node {$p$};
\draw (q) +(-0.25,0) node {$q$};
\draw (r) +(-0.25,0) node {$r$};
\draw (s) +(-0.25,0) node {$s$};
\end{tikzpicture}
\caption{\label{figure:xpaths} A graph~$G$ and an odd cycle transversal~$X$. Suppose~$M=\{\{r,s\}\}$. The dashed path~$P_1$ is an important~$p-q$~$X$-path. The dashed path~$P_2$ is a non-important~$r-s$~$X$-path. Further, the two vertices marked by gray boxes intersect all important~$X$-paths.}
\end{figure}
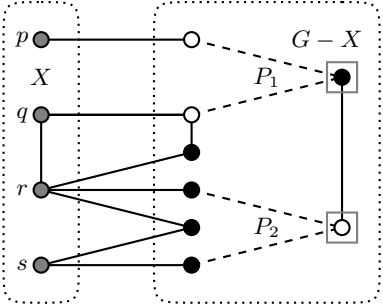}

\xpathfigure

\newcommand{\twocoloringextensionproof}{
\begin{proof}
We use the following procedure to try to extend~$c$ to a proper~$2$-coloring of~$G$. Pick a colored vertex~$u$ that has an uncolored neighbor~$v$ and define~$c(v):=1-c(u)$. For constructing the claimed path, let us orient the edge~$\{u,v\}$ from~$u$ to~$v$. Repeat until no colored vertex has an uncolored neighbor. (We will later color the vertices that are not connected to a vertex of~$S$.)

We observe that every colored vertex has a unique directed path from~$S$ to itself (since each vertex has at most one incoming edge) and that all directed edges~$(u,v)$ we have~$c(u)\neq c(v)$.

Let us assume that this extension of the coloring fails and that there is an edge~$\{u,v\}$ with~$c(u)=c(v)$. Since we extended a proper coloring for~$G[S]$ it follows that at least one of the vertices~$u$ and~$v$ is not in~$S$; w.l.o.g.~$v\in V(G) \setminus S$. Let~$P_u$ and~$P_v$ be the unique paths that connect~$S$ to~$u$ and~$v$ (note that possibly~$P_u=(u)$ in the case that~$u\in S$). We note that all edges along the two paths are properly colored and that all vertices except for the starting vertices in~$S$ are from~$V(G) \setminus S$. Furthermore, the two paths can at most overlap in their endpoints in~$S$: indeed, if they share a vertex~$w\in V(G) \setminus S$, then it can be easily seen that~$G-S$ must contain an odd cycle, contradicting the assumption that~$G-S$ is bipartite (the subpaths from~$w$ to~$u$ and~$v$ are contained in~$G-S$ and they are properly colored).

Thus concatenating the two paths via the edge~$\{u,v\}$ we obtain a path~$P$ between two vertices~$p,q\in S$ which is simple except possibly for~$p=q$. It can be easily seen that one of the two conditions on the parity of the length of~$P$ must hold. E.g., if the path has~$2t+1$ internal vertices, then it has~$2t+1$ directed edges and as well as the edge~$\{u,v\}$. This corresponds to a chain of~$2t+1$ disequalities and one equality (recall that~$c(u)=c(v)$) between the colors of~$p$ and~$q$, giving~$c(p)\neq c(q)$.

If the procedure succeeds then we obtain a proper~$2$-coloring of all connected components of~$G$ that intersect~$S$. By assumption all other components are bipartite and we can efficiently find proper $2$-colorings for them. This completes the proof.
\qed
\end{proof}}

\twocoloringextensionproof

\subsection{Proof of Lemma~\ref{hittingsetandannotations}}

\newcommand{\hittingsetandannotationsproof}{
\begin{proof}
Let~$V(G) \setminus X=P\cup Q$ be a bipartition of~$G-X$.
We point out that any odd~$X$-path between two vertices~$u,v\in X$ must be of the form~$(u,x,\ldots,y,v)$ where either~$x,y\in P$ or~$x,y\in Q$, since~$G-X$ is bipartite. Similarly, for any even~$X$-path~$(u,x,\ldots,y,v)$ we must have either~$x\in P$ and~$y\in Q$ or~$x\in Q$ and~$y\in P$.

We call Algorithm~\ref{algorithm:computehittingset} as ComputeHittingSet$(G,X,\ell)$. It computes for each pair~$\{u,v\}\in \binom{X}{2}$ vertex sets~$PP$,~$QQ$,~$PQ$, and~$QP$ intersecting all~$u-v$~$X$-paths that enter and leave~$G-X$ in the respective side of the bipartition (e.g.,~$PP$ intersects all paths~$(u,x,\ldots,y,v)$ with~$x,y\in P$).

Clearly, if~$PP>\ell$ or~$QQ>\ell$ then there are more than~$\ell$ odd~$u-v$~$X$-paths and~$\{u,v\}$ is correctly added to~$B$: we rely on the fact that by Menger's theorem, the maximum number of internally vertex-disjoint~$u-v$ paths equals the size of a minimum~$u-v$ vertex-cut. Otherwise,~$PP\cup QQ$ is a set of at most~$2\ell$ vertices that intersects all odd~$u-v$~$X$-paths. The above observation that odd~$u-v$~$X$ paths can only enter and leave~$G-X$ in these two ways is crucial here. The analog, for even paths, is true when~$PQ>\ell$ or~$QP>\ell$.

Then, for all vertices~$v\in X$, sets~$PQ$ are computed in the same way. It is easy to see that the same argumentation applies there and that computation of sets~$QP$ is not necessary. Note that each such path is indeed a~$p-p$~$X$-path since it enters and leaves~$G-X$ in different vertices. Clearly, the total size of~$H$ is bounded by~$4\ell\cdot \left(\binom{|X|}{2}+|X|\right)\leq 4\ell\cdot|X|^2$.

Finally, it is easy to see that this computation can be performed in polynomial time. The main work lies in the subroutine calls to VertexCut$(G,P,Q;u,v,S,T)$ (with~$S,T\in\{P,Q\}$) which can be implemented as follows: Make an auxiliary graph~$H$ by taking~$G[P\cup Q]$ and adding a source~$s$ adjacent to all neighbors of~$u$ in~$S$ and a sink~$t$ adjacent to the neighbors of~$v$ in~$T$. Then use a polynomial-time algorithm to compute a maximum set of internally vertex-disjoint~$s-t$ paths in~$H$. Schrijver~\cite[Theorem 9.3]{Schrijver03} gives an~$\BigO(nm)$ time algorithm providing a maximum packing of paths as well as a corresponding vertex-cut of the same cardinality.
\qed
\end{proof}}

\hittingsetandannotationsproof

\subsection{Proof of Lemma~\ref{lemma:makeannotations}}

\newcommand{\makeannotationsproof}{
\begin{proof}
We start from the obviously equivalent instance~$(G,X,M,\ell)$ with~$M=\emptyset$. Each transformation can be easily seen to be correct.

First, we consider single vertices~$p\in X$. If~$p\in C$ then there are more than~$\ell$ even length~$p-p$~$X$-paths in~$G$ that are vertex-disjoint (by definition~$p$ is not part of the~$X$-path). Hence there are more than~$\ell$ odd cycles in~$G$ which pairwise intersect only in~$p$, and therefore any odd cycle transversal of~$G$ of size at most~$\ell$ must contain~$p$. Therefore, we may delete~$p$ from~$G$ and decrease~$\ell$ by~$1$. This does not affect~$X$-paths between other vertices since~$p\notin V(G) \setminus X$. 

Next, we consider pairs of vertices~$\{p,q\}\in\binom{X}{2}$. If~$\{p,q\}\in A$ then there are more than~$\ell$ even length~$X$-paths. Hence, no deletion of at most~$\ell$ vertices can remove all those paths implying that whenever there is a set~$S \subseteq V(G)$ of size at most~$\ell$ such that~$G - S$ is bipartite, then if~$p, q \not \in S$ then vertices~$p,q$ must receive different colors in any proper $2$-coloring of~$G - S$, since at least one~$p-q$ path with an even number of vertices will not be intersected by~$S$: hence we may add an edge between~$p$ and~$q$ without changing the instance. Similarly, if~$\{p,q\}\in B$ then we may add the annotation~$\{p,q\}$ to~$M$.

The set~$X'$ is simply what is left of~$X$ after the vertex deletions. It is easy to see that all important~$p-q$~$X'$-paths in the obtained instance must be such that~$p,q\notin C$ and (depending on their parity)~$\{p,q\}\notin A$ or~$\{p,q\}\notin B$, since we added annotations respectively deleted the vertices in the other cases. Hence,~$H$ is a hitting set for all important~$X'$-paths.
\qed
\end{proof}}

\makeannotationsproof

\subsection{Proof of Lemma~\ref{protrusionDecomposition}}

\newcommand{\protrusiondecompositionproof}{
\begin{proof}
Let~\T be rooted at an arbitrary vertex. We will mark bags (nodes) of~\T to select a set~$S'$. First, for every~$v\in S$ we mark a bag of~\T that contains~$v$ (i.e., at most~$|S|$ bags).

Second, for any two marked bags, we also (exhaustively) mark their lowest common ancestor in~\T by the following procedure. Let all bags marked so far be \emph{active}; there are at most~$|S|$ such bags. Identify the lowest bag, say~$B$, that is a common ancestor of at least two active bags, say~$B_1$ and~$B_2$. Mark~$B$ (if it was unmarked) and set it active. Furthermore, set all other marked bags in the subtree rooted at~$B$ to \emph{inactive} (this includes~$B_1$ and~$B_2$).
Observe that any lowest common ancestor of a bag in the subtree and some bag~$B'$ outside the subtree is also a lowest common ancestor of~$B$ and~$B'$. Hence it suffices to proceed only for the active bags. Since the number of active bags is reduced by at least one each time that another bag is marked, we mark at most~$|S|$ additional bags.

Let~\B denote the set of marked bags; clearly~$|\B|\leq 2|S|$. Now, let~$S'$ denote the set of all vertices that are in a marked bag of~\T; a total of at most~$2|S|(w+1)$ vertices. Clearly~$S\subseteq S'$ since we marked a bag for each~$v\in S$. To establish the lemma it remains to prove that the number of neighbors that a connected component of~$G - S'$ has in the set~$S'$ is appropriately bounded.

So let~$C$ be an arbitrary connected component of~$G-S'$ and let~$s \in S'$ be a neighbor of~$C$ in~$S'$. There must be a connected component~$\T_C$ of~$\T-\B$ that contains all vertices of~$C$. Furthermore,~$s$ must be contained in at least one bag of~$\T_C$. Since~$s$ is also contained in at least one marked bag, there must be a marked bag that is adjacent to~$\T_C$ in~$\T$ which contains~$s$. The reason is that all occurrences of a vertex in bags of~$\T$ must be connected and that all bags adjacent to~$\T_C$ are marked (as~$\T_C$ is a connected component of~$\T-\B$). This argument shows that all neighbors of~$C$ in the set~$S'$ must occur in marked bags adjacent to~$\T_C$. To be able to bound the number of such neighbors, we show that the number of adjacent marked bags is at most two.

Let us assume for contradiction that~$\T_C$ is adjacent to at least three marked bags. It follows that at least two of those bags are children of~$\T_C$ (i.e., they are adjacent to~$\T_C$ in~\T and they are below~$\T_C$ with respect to the root of~\T). This, however, would imply that~$\T_C$ must contain the lowest common ancestor of two marked bags; a contradiction.

Hence~$\T_C$ is adjacent to at most two marked bags and, therefore,~$C$ can have at most~$2(w+1)$ neighbors in~$S'\supseteq S\cup S'$.

For a bound of at most~$2w$ neighbors consider the following: a neighbor~$s$ of~$C$ must be in a marked bag, say~$B\in\B$, adjacent to~$\T_C$, but it must also be contained in a bag together with a vertex of~$C$. Let~$B'$ denote a bag of~$\T_C$ that contains vertices of~$C$ and that is nearest to~$B$ (i.e., adjacent to~$B$ or connected to it by a path of bags that contain no vertices of~$C$). This bag is unique since bags containing vertices of~$C$ form a subtree of~$\T_C$. It is easy to see that all~$S'$-neighbors of~$C$ which are in~$B$ must also be contained in~$B'$. That bag, however, must also contain at least one vertex of~$C$. Hence, each marked bag adjacent to~$\T_C$ can contribute at most~$w$ neighbors, and~$C$ has at most~$2w$ neighbors in~$S'=S\cup S'$.
\qed
\end{proof}}

\protrusiondecompositionproof

\subsection{Proof of Lemma~\ref{lemma:numberofcomponents}}

\newcommand{\numberofcomponentsproof}{
\begin{proof}
Let~$\C$ denote the set of connected components of~$(G-X)-H$. To identify connected components in~$\C$ that may be safely deleted, we use a similar but simpler procedure as in the proof of Lemma~\ref{hittingsetandannotations}.

The main idea is that the way in which a component~$C \in \C$ affects the problem instance is by possibly providing a path between two vertices~$p,q \in X \cup H$: the coloring implications along this path might prevent some colorings of~$X \cup H$ from being valid, and in a solution we might want to delete a vertex from~$C$ to break this path of implications. But if there are more than~$\ell$ components which provide a path of the same parity between~$p$ and~$q$, then we cannot destroy all such paths by $\ell$ vertex deletions, and hence this fixes the relative colors of~$p$ and~$q$ in every solution. As soon as there are more than~$\ell$ components which provide a path of a given parity for some pair~$p,q \in X \cup H$, the existence of additional components which realize the same path is not relevant to the problem anymore, and we can remove such components if they are not relevant to \emph{any} pair~$p,q \in X \cup H$. In the remainder of the proof we formalize this idea into a reduction procedure.

For each pair of vertices~$p,q\in X\cup H$ (also for~$p=q$) and a choice of odd or even parity, we test for each component~$C \in \C$ whether there is a path from~$p$ to~$q$ whose internal vertices lie in~$C$ and for which the number of internal vertices matches the chosen parity. We can perform this test by~$2$-coloring the component~$C$ (which is bipartite since it is a subgraph of~$G - X$), observing that~$C$ provides a~$p-q$ path of odd (resp.\ even) parity if and only if~$p$ and~$q$ have neighbors of the same color (resp.\ different colors) in the component. For the given choice of~$p,q$ and given parity, we mark the first~$\ell+1$ components of~$\C$ that provide an appropriate path.

After doing this for all pairs, we delete all unmarked components of~$\C$ from~$G$, obtaining~$G'$.
Clearly, we have marked at most~$2\cdot(\ell+1)\cdot(|X|+|H|)^2$ components, and only those exist also in~$(G'-X)-H$.

Let us argue equivalence of the two instances. Clearly, deleting vertices of~$G$ can only make the problem easier, so assume for contradiction that~$(G',X,M,\ell)$ is \yes, but that~$(G,X,M,\ell)$ is \no. Accordingly, let~$S\subseteq V(G')$ be a set of at most~$\ell$ vertices of~$G'$ such that~$G'-S$ is bipartite and let~$c' \colon V(G' - S)\to\{0,1\}$ be a proper~$2$-coloring of~$G'-S$ which respects the annotations. We will show how to extend~$c'$ to a proper~$2$-coloring of~$G-S$, proving that~$(G,X,M,\ell)$ is \yes too.

We start from a partial $2$-coloring~$c$ of~$G - S$ which is obtained by restricting~$c'$ to the vertices of~$X \cup H$. We will show how to extend~$c$ to a $2$-coloring of the entire graph~$G - S$. Note that~$(G - X) - H$ is bipartite (since it is a subgraph of the bipartite graph~$G - X$) and hence we may apply \lemmaref{lemma:twocloringextension} to the graph~$G - S$ with the $2$-coloring~$c$ of~$G[X \cup H]$, letting~$X \cup H$ play the role of the set~$S$ in the statement of \lemmaref{lemma:twocloringextension}. By the lemma we either find an extension of~$c$ to a proper $2$-coloring of the entire graph~$G - S$ (and we would be done, since this $2$-coloring of~$G - X$ must respect the annotations since~$c'$ does), or we find a path~$P$ between two vertices~$p, q \in X \cup H$ such that the parity of this path conflicts with the colors assigned to~$p,q$ by function~$c$. We will show that this latter case leads to a contradiction, and that therefore we must always be able to extend to a proper $2$-coloring. By the guarantee of the lemma, all internal vertices on~$P$ are from the set~$(G - S) - (X \cup H)$ and hence the internal vertices of~$P$ are all contained within a single connected component~$C \in \C$.

It is not hard to see that if~$G' - S$ can be properly $2$-colored with the given colors for vertices~$p$ and~$q$, then there can be no path between~$p$ and~$q$ in the graph~$G' - S$ whose parity equals the parity of~$P$. So in particular, the component~$C$ cannot exist in graph~$G'$ and must have been deleted when forming~$G$. But by the definition of the reduction procedure, if the component~$C$ was deleted it was not marked, and hence we marked~$\ell + 1$ components which provided a~$p-q$ path of the same parity. Since at least one of these components is not intersected by~$S$ (which has size at most~$\ell$), this shows that~$G' - S$ must contain a $p-q$ path of the same parity as~$P$; a contradiction to the assumption that~$c'$ properly $2$-colors the graph~$G' - S$.

This contradiction shows that when applying \lemmaref{lemma:twocloringextension} we must always obtain a proper $2$-coloring~$c$ of~$G - S$, and since~$c$ assigns the same colors to~$X \cup H$ as the function~$c'$ which respects the annotations, we find that~$c$ is a $2$-coloring of~$G - S$ which respects the annotations; this proves that~$(G,X,M,\ell)$ is a \yes instance and completes the proof.
\qed
\end{proof}}

\numberofcomponentsproof

\subsection{Proof of Lemma~\ref{lemma:separatorreplacement}}

\newcommand{\separatorreplacementproof}{
\begin{proof}
Let~$c \colon V(G-R) \to \{0,1\}$ be a proper~$2$-coloring of~$G-R$ that respects the annotations on~$X$. Let a partial coloring~$c' \colon V(G-R')\to\{0,1\}$ of~$G-R'$ be defined via:~$c'(v):=c(v)$ for all vertices that are not in the component~$C$ of~$(G-X)-H$. Note that~$R$ and~$R'$ differ only on these vertices and, hence,~$(G-V(C))-R=(G-V(C))-R'$. Thus~$c'$ is a proper~$2$-coloring of~$(G-V(C))-R'$ and respects the annotations, since~$X\subseteq V(G)\setminus V(C)$ and~$c$ respects the annotations.

We apply Lemma~\ref{lemma:twocloringextension} on the graph~$G-R'$, the coloring~$c'$ and using the vertex set~$V(G)\setminus (V(C)\cup R')$ as the set~$S$ in the lemma. Let us assume for contradiction that we obtain a connected component~$C' \subseteq C$ and a simple path~$P$ between two vertices~$p,q\in N_{G-R'}(C')\subseteq X\cup H$ whose internal vertices are from~$V(C')$, with the guarantee that~$P$ cannot be properly~$2$-colored given the colors of~$p$ and~$q$. We will derive a contradiction by a case analysis on the status of~$p$ and~$q$.


$\boldsymbol{(i): p,q\in X:}$ We first consider the case that both endpoints of the path are contained in~$X$. \lemmaref{lemma:twocloringextension} guarantees that either~$p \neq q$ and~$P$ is a simple path, or that~$p = q$ and~$P$ is an odd cycle through~$p = q$. Let~$P'$ be the interior of the path: $P'$ is obtained from~$P$ by deleting~$p$ and~$q$, and it is not hard to verify that~$P'$ cannot be empty if~$c$ is a proper coloring. Clearly,~$P'$ is an~$X$-path between~$p$ and~$q$. Since~$H$ intersects all important~$X$-paths and all vertices of~$P'$ lie in~$C' \subseteq C$ which is a component of~$(G - X) - H$, it is clear that~$P'$ is not an important $X$-path. This, however, implies that there must be either an edge~$\{p,q\}$ in~$G$ or an annotation~$\{p,q\}$ in~$M$ and that the path~$P'$ has a matching length. But~$P'$ cannot be properly colored given the colors of~$p$ and~$q$, whereas the fact that~$P'$ is not important implies that whenever the endpoints are colored according to the annotations, the path \emph{can} be properly colored: this implies that~$c$ does not respect the annotations, which is a contradiction.

$\boldsymbol{(ii): p,q\in T:}$ In this case the vertices~$p$ and~$q$ are terminals of the labeled graph $(D, X \cup T, f)$. If~$p = q$ then by \lemmaref{lemma:twocloringextension} we know~$P$ is an odd cycle through~$p = q$, and since~$P$ is entirely contained within~$D$ (which is a subgraph of~$G - X$) this contradicts the assumption that~$G - X$ is bipartite. Hence in the remainder of this case we assume that~$p \neq q$. Since~$P$ is a~$p-q$ path in~$G - R'$ whose internal vertices are from~$C$, and since~$S' = V(C) \cap R'$, the set~$S'$ does not separate~$p$ from~$q$ in~$C$. It follows from the definition of the labeling function~$f$ and \defref{reachableLabelsDef} and \defref{cutCharacteristicDef} that we must have~$p \in \L(q, S')$ and~$q \in \L(p, S')$. Since the cut characteristics of~$S$ and~$S'$ with respect to this labeled graph are the same by the assumption that~$\K(S, T) = \K(S', T)$, we must have~$p \in \L(q, S)$ and~$q \in \L(p, S)$ which shows that vertices~$p$ and~$q$ are also connected in~$D - S$. If we take a path~$P'$ from~$p$ to~$q$ in the graph~$D - S$, then~$P'$ must also be a~$p-q$ path in the graph~$G - R$. Now, if the parities of~$P$ and~$P'$ differ then~$P \cup P' \subseteq V(C) \cup T$ must contain an odd cycle, which contradicts the assumption that~$G - X$ is bipartite. If the parities are the same, then the fact that~$P$ cannot be properly $2$-colored given the colors of~$p, q$ implies that~$P'$ cannot be properly $2$-colored, which contradicts the assumption that~$c$ is a proper $2$-coloring of~$G - R$.


$\boldsymbol{(iii): \mbox{\textbf{w.l.o.g.}\ } p\in T, q\in X:}$ Observe that the requirements for this case imply~$p \neq q$ and hence~$P$ is a~$p-q$ path in~$G - R'$ whose internal vertices are contained in~$C$: hence the subpath~$P - \{q\}$ also exists in the graph~$D - S'$ and connects~$p$ to a neighbor of~$q$, showing by the definition of the labeling function that~$q \in \L(p, S')$. Since~$S$ and~$S'$ have the same cut characteristic with respect to the labeled graph~$(D, X \cup T, f)$ we must have~$q \in \L(p, S)$: this implies that in~$D - S$ there is a path from vertex~$p$ to a vertex labeled~$q$, and this vertex labeled~$q$ must be a neighbor to~$q$ in the graph~$G$. Hence there is a path~$P'$ from~$p$ to~$q$ through the component~$C$ in the graph~$G - R$. Since~$c$ is a proper $2$-coloring of~$G - R$, it must properly $2$-color the path~$P'$. Since~$P$ cannot be properly $2$-colored given the colors of~$p$ and~$q$, it follows that the paths~$P$ and~$P'$ must have different parities. Since~$P$ and~$P'$ are two paths of different parities between distinct vertices~$p$ and~$q$, their union must contain an odd cycle. Since~$G - X$ is bipartite by assumption, the union of~$P$ and~$P'$ must contain an odd cycle~$Q$ through a vertex in~$X$, and hence this odd cycle must intersect~$q$. Let~$x,y$ be the predecessor and successor to~$q$ on the odd cycle~$Q$; it is easy to see that~$x,y \in V(C)$. It follows that~$Q - \{q\}$ is a path between~$x$ and~$y$ in~$G - X$ containing an even number of vertices. Since~$C$ is a connected component containing~$x$ and~$y$, there is a simple path~$\hat{P}$ from~$x$ to~$y$ which only uses vertices from~$C$. Since~$G - X$ is bipartite all simple paths between two given vertices in~$G - X$ must have the same parity, which shows in particular that~$\hat{P}$ must contain an even number of vertices since~$Q - \{q\}$ is a path between~$x$ and~$y$ in~$G - X$ with an even number of vertices. Now~$\hat{P}$ is entirely contained within~$C$, and~$\hat{P}$ forms an odd cycle with~$q$; but by \defref{importantXPath} this implies that~$\hat{P}$ is an important~$q-q$ $X$-path, contradicting the assumption that~$H$ intersects all important $X$-paths. This concludes the proof of this last case.

Thus, in all cases we have found a contradiction. This implies that the application of Lemma~\ref{lemma:twocloringextension} must provide a proper~$2$-coloring of~$G-R'$ that is an extension of~$c'$. Thus,~$G-R'$ has a proper~$2$-coloring that respects the annotations and, hence,~$R'$ is also a valid solution, as claimed.
\qed
\end{proof}}

\separatorreplacementproof

\subsection{Proof of Lemma~\ref{lemma:restrictedproblem}}

\newcommand{\restrictedproblemproof}
{
\begin{proof}
The proof is organized as follows. First, for each component~$C$ of~$(G-X)-H$ we will partition all separators of size at most~$\delta-1$ into equivalence classes according to their cut characteristic with respect to a labeled graph whose labels express adjacency to~$X$ and to the neighborhood of the component; accordingly these separators are subsets of~$V(C)\cup(N(C)\cap H)$. We will arbitrarily pick one minimum-size representative for each class and mark its vertices in~$C$ as deletable. All other vertices of the component~$C$ will remain undeletable. Doing this for all components we obtain an instance of the restricted annotated problem. Second, we will show that for each odd cycle transversal~$R$ of~$G$ which allows a $2$-coloring of~$G - R$ respecting the annotations, there is a transversal~$R'$ of at most the same size that intersects each component of~$(G-X)-H$ in deletable vertices (i.e., in a representative separator). From this, equivalence of the instance of the restricted annotated problem follows immediately.

For each component~$C$ of~$(G-X)-H$ we define the necessary vertex sets to express the cut characteristics of its separators and to be able to invoke Lemma~\ref{lemma:separatorreplacement} later on. We will omit subscripts~$C$ for readability and always focus only on one component at a time. Let~$T=\{t_1,\ldots,t_{\delta'}\}=N_G(C)\setminus X$ be the set of the~$\delta'\leq\delta$ vertices of~$H$ that are adjacent to~$C$. We define a labeled graph~$(D, X \cup T, f)$ on the base graph~$D := G[V(C)\cup T]$. Its vertices are labeled by~$f \colon V(D)\to X \cup T$ (exactly as in Lemma~\ref{lemma:separatorreplacement}), i.e., each vertex is labeled by its set of neighbors in~$X$ plus possibly by itself if it is in~$T$:
\begin{align*}
f(v) := 
\begin{cases}
N_G(v) \cap X & \mbox{If~$v \not \in T$.} \\
(N_G(v) \cap X) \cup \{v\} & \mbox{If~$v \in T$.}
\end{cases}
\end{align*}

Now, we consider all separators of size at most~$\delta-1$, i.e., all~$S\in\binom{V(C)\cup T}{\leq \delta - 1}$. We let two such separators~$S$ and~$S'$ be equivalent, if they have the same cut characteristic in the labeled graph~$(D,X\cup T,f)$ with respect to the terminals~$T$, i.e., if~$\K(S,T)=\K(S',T)$. It can be easily checked that a partition of~$\binom{V(C)\cup T}{\leq \delta - 1}$ into equivalence classes can be computed in time polynomial in~$\binom{|V(C)|+|T|}{\leq \delta - 1}+|V(G)|$: the main work consists of running~$|T|$ breadth-first searches for each separator to identify the set of reachable labels. Since we assume~$\delta$ to be a constant, this takes polynomial time in the input size.

We arbitrarily pick one separator of minimum size as the representative for each class. Each equivalence class is characterized by a tuple in~$\K^{\delta - 1}(T)$ with respect to the labeled graph~$(D, X \cup T, f)$. By \thmref{simpleCutCharacteristicBound} the number of equivalence classes is therefore bounded by~$\kappa(\delta, \delta - 1, |X| + \delta)$ since we are considering labeled graphs with~$n = \delta$ terminals, for which we look at separators of size at most~$m = \delta - 1$ in a graph with~$|X| + \delta' \leq |X| + \delta$ different labels. We now mark as deletable all vertices of~$C$ which occur in a representative separator (of size at most~$\delta - 1)$. Thus, per component, we mark less than~$\delta$ vertices for each representative separator, resulting in a total number of at most~$\delta \cdot \kappa(\delta, \delta - 1, |X| + \delta)$ marked vertices per component~$C$.

After doing this for all~$\alpha$ components there are at most~$\alpha \cdot \delta \cdot \kappa(\delta, \delta - 1, |X| + \delta)$ vertices marked as deletable in~$(G - X) - H$. In addition we mark the vertices of~$X$ and~$H$ as deletable. We now let~$Z$ contain the marked vertices, from which it easily follows that~$Z$ does not exceed the stated size bound. It remains to prove that the instance~$(G,X,M,\ell)$ of \annotatedbiptwoct is equivalent to instance~$(G,X,M,\ell,Z)$ of \restrictedbiptwoct. This equivalence will follow mainly from the following claim.

\begin{claim}
Let~$R \subseteq V(G)$ be a set of vertices such that~$G-R$ is bipartite and has a proper~$2$-coloring that respects the annotations. Then there is another set~$R' \subseteq Z$ of deletable vertices of size~$|R'| \leq |R|$ such that~$G-R'$ is bipartite and and has a proper~$2$-coloring that respects the annotations.
\end{claim}

\begin{claimproof}
Let~$R$ be an odd cycle transversal of~$G$ such that~$G - R$ has a $2$-coloring respecting the annotations. If~$R \subseteq Z$ we are done, so assume that~$R$ contains at least one undeletable vertex and let~$C$ be the connected component of~$(G-X)-H$ containing the vertex. Let~$T$ contain the neighbors that~$C$ has in~$H$ in some arbitrary order, i.e.,~$T:=N_G(C)\cap H$. Let~$S:=R\cap (V(C)\cup T)$ be the intersection of~$R$ with~$C$ augmented by its neighborhood in~$H$.

We will show that we can obtain a solution~$R'$ which is not bigger than~$R$, and which only intersects~$C$ in deletable vertices; we will then appeal to induction to show that this process can be repeated to obtain a solution which is a subset of~$Z$. To find a set~$R'$ such that~$R' \cap V(C) \subseteq Z$ we distinguish between two cases.

\begin{itemize}
	\item We first consider the case that~$|S| \leq \delta - 1$. Let~$S'$ be the representative separator that has the same cut characteristic as~$S$ with respect to the labeled graph~$(D,X\cup T,f)$ (as defined above). Since we remembered a \emph{minimum-size} representative of each class, we must have~$|S'| \leq |S|$. We now argue that~$S\cap T=S'\cap T$. Assume there is a terminal~$t_i \in T \setminus S$. Since this terminal is labeled with itself, we have~$t_i \in f(t_i)$ and since~$t_i \not \in S$ we must have~$t_i \in \L(t_i, S)$. Since the cut characteristics of~$S$ and~$S'$ are identical, we must have~$t_i \in \L(t_i, S')$ and it follows directly from \defref{reachableLabelsDef} this is only possible if~$t_i \not \in S'$. Similarly for every~$t_i \in T \cap S$ we have~$t_i \not \in \L(t_i, S)$ which is only possible if~$t_i \not \in \L(t_i, S')$. Hence the fact that~$S \cap T = S' \cap T$ is implied by the fact that~$S$ and~$S'$ have the same cut characteristic.
	
Now, from applying Lemma~\ref{lemma:separatorreplacement} for~$S\setminus T$ and~$S'\setminus T$ (both subsets of~$V(C)$) we get that
\[
R'=(R\setminus (S\setminus T))\cup (S'\setminus T)=(R\setminus S)\cup S'
\]
is also an odd cycle transversal of~$G$ such that~$G - R'$ has a proper $2$-coloring respecting the annotations, and since~$S'$ is a representative separator whose vertices were marked as deletable we have~$R' \cap V(C) = S' \subseteq Z$.

\item If~$|S|\geq \delta$ then we replace it by the set~$T$, i.e.,~$R'=(R\setminus S)\cup T$, implying that~$|R'|\leq|R|$. Let us briefly argue that~$R'$ is indeed a valid solution. Let~$c$ be a proper~$2$-coloring of~$G-R$ that respects the annotations, and create a~$2$-coloring of~$R'$ by first copying the coloring of~$c$ onto all vertices outside of~$V(C)\cup T$. Since all~$X$-paths through~$C$ (not crossing vertices of~$H$) must match annotations (and edges) among vertices of~$X$, it follows from Lemma~\ref{lemma:twocloringextension} that a greedy extension of the~$2$-coloring of~$X$ into~$C$ cannot fail (e.g., see also the argument in the proof of Lemma~\ref{lemma:numberofcomponents} and note that all neighbors of~$C$ in~$G-R'$ are in~$X$).
\end{itemize}

In both cases we have identified an odd cycle transversal~$R'$ of at most the same size as~$R$ that only intersects~$C$ in deletable vertices, and which ensures that~$G - R'$ has a $2$-coloring respecting the annotations. Since this replacement step within the component~$C$ does not affect the intersection of the solution with any other components of~$(G - X) - H$, we may repeatedly apply such replacement steps until we arrive at a solution which is entirely contained within~$Z$. This completes the proof of the claim.
\end{claimproof}

The given claim allows us to prove that the restricted annotated instance is equivalent tot he annotated instance. Clearly, restricting the set of deletable vertices can only make the problem harder: if the restricted instance is \yes then the same deletion set is a valid solution to the original. For the other direction it follows from the claim that if the original is \yes, then the restricted instance is also \yes. This completes the proof of the lemma.
\qed
\end{proof}}

\restrictedproblemproof

\subsection{Proof of Lemma~\ref{lemma:backtransformation}}

\newcommand{\backtransformationproof}{
\begin{proof}
Let~$(G,X,M,\ell,Z)$ be an instance of the restricted annotated problem. Clearly, if~$G-Z$ is not bipartite, then the instance is \no and we may return a dummy \no instance of \biptwoct.

We argue that we may assume w.l.o.g.\ that~$M=\emptyset$: If~$M$ is not empty, then for each~$\{u,v\}\in M\subseteq\binom{V}{2}$ we may add a new vertex~$w$ adjacent to~$u$ and~$v$. Since~$w$ is not in~$Z$ it may not be deleted. Hence, for any odd cycle transversal~$S\subseteq Z$, either~$u\in S$ or~$v\in S$ or~$u$ and~$v$ must have the same color in any~$2$-coloring of~$G-S$. This does not affect~$\ell$ or~$Z$. Henceforth, we assume~$M=\emptyset$.

We construct the graph~$G'$, starting from~$G'=G[Z]$. For all~$p,q\in Z$ we do the following:
\begin{itemize}
\item If there is an odd~$Z$-path between~$p$ and~$q$ in~$G$, i.e., a~$p-q$ path with internal vertices only from~$V(G)\setminus Z$, then add~$\ell+1$ new vertices to~$G'$ and make them adjacent to both~$p$ and~$q$.
\item If there is an even~$Z$-path between~$p$ and~$q$ in~$G$, then add the edge~$\{p,q\}$ (unless it exists already in~$G[Z]$).
\end{itemize}
Recall, that the existence of odd and even~$p-q$~$Z$-paths can be easily checked by~$2$-coloring~$G-Z$ and checking whether~$p$ and~$q$ have neighbors of the same, respectively, different colors in some component of~$G-Z$.

Finally, we let~$X':=Z$ and return the instance~$(G',X',\ell)$. Clearly,~$G'-X'$ is an independent set, and hence it is bipartite and has bounded treewidth, since we only added vertices that are adjacent to~$Z$ (but not to one another).

\begin{claim}
$(G,X,M,\ell,Z)$ is \yes if and only if~$(G',X',\ell)$ is \yes.
\end{claim}

\begin{claimproof}
$\boldsymbol{(\Rightarrow):}$ Assume first that~$(G,X,M,\ell,Z)$ is \yes and let~$S\subseteq Z$ be an odd cycle transversal of~$G$ of size at most~$\ell$ (recall that~$M=\emptyset$). To see that~$S$ is also an odd cycle transversal of~$G'$, let~$c \colon V(G-S)\to\{0,1\}$ be a proper~$2$-coloring of~$G-S$. We will show how the restriction of~$c$ to~$Z$ can be extended to a proper~$2$-coloring of~$G'-S$. Recall, that~$G'$ is the same as~$G[Z]$ except possibly for additional edges between vertices of~$Z$ and additional vertices that are adjacent to pairs of vertices from~$Z$.

Let us check first that the additional edges pose no problem: if~$p,q\in Z\setminus S$ and~$\{p,q\}$ is an edge of~$G'$ but not of~$G$ then there must be an even~$Z$-path between~$p$ and~$q$ in~$G$. This path exists in~$G-S$, since~$p,q\notin S$ and~$S\subseteq Z$. Hence,~$c(p)\neq c(q)$.

Now, let us consider the additional vertices: if~$p,q\in Z\setminus S$ have a shared neighbor in~$G'$ then there must be an odd~$Z$-path between~$p$ and~$q$ in~$G$. Again this path must exist also in~$G-S$, implying that~$c(p)=c(q)$. Therefore, we may color such a shared neighbor with color~$1-c(p)$. It follows that~$S$ is an odd cycle transversal of~$G'$ and that~$(G',X',\ell)$ is \yes.

$\boldsymbol{(\Leftarrow):}$ Now, assume that~$(G',X',\ell)$ and let~$S'$ be an odd cycle transversal of~$G'$ of size at most~$\ell$. We let~$S:=S'\cap Z$ and claim that~$S$ is an odd cycle transversal of~$G$ (clearly~$S\subseteq Z$ and~$|S|\leq \ell$).

Let~$c' \colon V(G'-S')\to\{0,1\}$ be a proper~$2$-coloring of~$G'-S'$. We will extend its restriction to~$Z$ to a proper~$2$-coloring of~$G-S'$. Let us first note that~$G[Z]$ is a subgraph of~$G'[Z]$ and hence~$c'$ is a proper~$2$-coloring for~$G[Z]$. It is easy to see that a greedy extension of the coloring onto~$V(G-S')\setminus Z$ suffices: First of all,~$G-Z$ is bipartite, so connected components that do not intersect~$Z$ can be~$2$-colored arbitrarily. Second, if the coloring fails then there must be an odd path between two vertices~$p,q\in Z\setminus S$ with different colors under~$c'$, or an even path between vertices of the same color. In both cases, such a path led to adding an edge (for an even path) or~$\ell+1$ shared neighbors (for an odd path) between those vertices in~$G'$. In the first case, the two vertices must have different colors (as~$c'$ is proper for~$G'-S'[Z]=G'-S[Z]$). In the second case,~$S'$ cannot contain all shared neighbors, so the vertices must have the same color under~$c$. Thus, in both cases we find a contradiction, implying that~$G-S$ can be properly~$2$-colored and that~$(G,X,M,\ell,Z)$ is a \yes instance.

This completes the proof of the claim.
\end{claimproof}

It is easy to see that the construction can be performed in polynomial time. Correctness follows from the previous claim. To see that~$G'$ has at most~$|Z|+(\ell+1)\cdot|Z|^2$ vertices recall that the only additional vertices that we added to~$G[Z]$ are at most~$\ell+1$ shared neighbors per pair of vertices from~$Z$.
\qed
\end{proof}}

\backtransformationproof

\subsection{Proof of Theorem~\ref{theorem:polyKernelProof}}

\newcommand{\polykernelproof}{
\begin{proof}
We will sketch the actions of the algorithm for some fixed integer~$w$. Let~$(G, X, \ell)$ be an input of \biptwoct, and let~$k := |X|$ be the size of the parameter to the problem. Observe that we may assume without loss of generality that~$\ell < |X| = k$, otherwise the set~$X$ is an OCT of size at most~$\ell$ and therefore we can just output a constant-size \yes instance. This implies that we can assume that~$\ell + 1 \leq k$.

We first compute sets~$A, B \subseteq \binom{X}{2}$, a set~$C \subseteq X$ and a set~$H \subseteq V(G) \setminus X$ of size at most~$4 \ell |X|^2$ by \lemmaref{hittingsetandannotations}, and use these with \lemmaref{lemma:makeannotations} to obtain an instance $(G', X', M, \ell')$ of \annotatedbiptwoct with~$X' \subseteq X$ and~$\ell' \leq \ell$ such that~$H$ intersects all important~$X'$-paths of the annotated instance. We then use Bodlaender's algorithm~\cite{Bodlaender96} to compute a tree decomposition \T of the graph~$G' - X'$, which can be done in linear time since~$G' - X'$ has treewidth at most~$w$ which we treat as a fixed constant.

We now apply \lemmaref{protrusionDecomposition} to the triple~$(G' - X', \T, H)$: the set~$H$ plays the role of~$S$ in the lemma statement. We find a superset~$H' \supseteq H$ of size at most~$2(w+1)|H|$ such that for each connected component~$C$ of~$G' - X'$ it holds that~$|N_{G'}(C) \cap H'| \leq 2w$. Since~$H$ intersects all important~$X'$-paths of the instance~$(G', X', M, \ell')$ it is easy to see that the superset~$H'$ must also have this property.

We now apply \lemmaref{lemma:numberofcomponents} to the instance~$(G', X', M, \ell')$ and set~$H'$ to obtain in polynomial time an equivalent instance~$(G'', X', M, \ell')$ of \annotatedbiptwoct with the guarantee that~$(G'' - X') - H'$ has at most~$2(\ell' + 1)(|X'| + |H'|)^2$ connected components; observe that the sets~$X',M$ and the value of~$\ell'$ remains unchanged by this step. Since the graph~$G''$ is an induced subgraph of~$G'$, it follows that~$H'$ is also a hitting set for the important~$X'$-paths in the graph~$G''$.

Define~$\alpha := 2(\ell' + 1)(|X'| + |H'|)^2$ and~$\delta := 2w$; we may then apply \lemmaref{lemma:restrictedproblem} to the instance~$(G'', X', M, \ell')$ and the set~$H'$ to obtain an equivalent instance $(G'', X', M, \ell', Z)$ of \restrictedbiptwoct where~$|Z|$ is bounded by~$|X'| + |H'| + \alpha \cdot \delta \cdot \kappa(\delta, \delta - 1, |X'| + \delta)$. 

As the final step of the kernelization we move from the instance of the restricted, annotated problem back to the original problem. We apply \lemmaref{lemma:backtransformation} to the instance~$(G'',X',M,\ell',Z)$ to obtain an equivalent instance~$(G^*, X^*, \ell^*)$ of the original problem, and the lemma guarantees that~$|V(G^*)| \leq |Z| + (\ell' + 1)\cdot |Z|^2$. If~$\ell' < 0$ then the original input is equivalent to an instance which asks for a set of negative size; hence the original input is \no, and we can output a constant-size \no instance. If~$\ell' = 0$ then we can decide the instance in polynomial time: we output \yes if and only if~$G^*$ is bipartite. In the remaining cases, the instance~$(G^*, X^*, \ell^*)$ is used as the output to the kernelization algorithm and we are guaranteed that~$\ell^* > 0$. It follows directly from the intermediate lemmas that this procedure takes polynomial time for each fixed~$w$, and that the output instance~$(G^*, X^*, \ell^*)$ is equivalent to the original input~$(G,X,\ell)$. It remains to prove that the size of the output instance is indeed bounded polynomially in the parameter to the input problem, which is~$k = |X|$. This is just a matter of formula manipulation using the facts we collected above.
{\allowdisplaybreaks
\begin{align}
|V(G^*)| &\leq |Z| + (\ell' + 1) \cdot |Z|^2 & \mbox{By \lemmaref{lemma:backtransformation}.} \\
&\leq 2k|Z|^2 & \mbox{$\ell' \leq \ell < k$.} \label{vBound} \\
\delta &= 2w & \mbox{By definition.} \\
|X'| &\leq |X| \leq k & \mbox{Since~$X' \subseteq X$.} \\
|H'| &\leq 2(w+1)|H| \leq 2 \delta |H| & \mbox{By \lemmaref{protrusionDecomposition}.} \label{hprimebound} \\
|H| &\leq 4 \ell |X|^2 \leq 4 \ell k^2 \leq 4 k^3 & \mbox{By \lemmaref{lemma:makeannotations}.} \label{hbound} \\
|H'| &\leq 8 \cdot \delta \cdot k^3 & \mbox{By \eqref{hprimebound}, \eqref{hbound}.} \label{hprimebound2} \displaybreak[0] \\
|X'| + |H'| &\leq k + 8 \cdot \delta \cdot k^3 \leq 9 \cdot \delta \cdot k^3 & \mbox{By \eqref{hprimebound2}.} \label{sumBound} \displaybreak[0] \\
\alpha &= 2(\ell' + 1)(|X'| + |H'|)^2 & \mbox{By definition.} \label{alphaDef} \\
&\leq 2k (|X'| + |H'|)^2 & \mbox{$\ell + 1 \leq k$.} \\
&\leq 2k (9 \cdot \delta \cdot k^3)^2 & \mbox{By \eqref{sumBound}.} \\
&\leq 162 \cdot k^7 \cdot \delta^2 & \mbox{Simplifying.} \label{alphaBound} \displaybreak[0] \\
|X'| + |H'| &\leq \alpha & \mbox{By \eqref{alphaDef}.} \label{sumBound2} \displaybreak[0] \\
|Z| &\leq |X'| + |H'| + \alpha \cdot \delta \cdot \kappa(\delta, \delta - 1, |X'| + \delta) & \mbox{By \lemmaref{lemma:restrictedproblem}.}\\
&\leq \alpha + \alpha \cdot \delta \cdot \kappa(\delta, \delta - 1, |X'| + \delta) & \mbox{By \eqref{sumBound2}.} \\
&\leq 2 \alpha \cdot \delta \cdot \kappa(\delta, \delta - 1, |X'| + \delta) & \mbox{$\delta, \kappa(\ldots) > 0$.} \\
&\leq 2 \alpha \cdot \delta \cdot \kappa(\delta, \delta, k + \delta) & \mbox{Def.\ of~$\kappa(\cdot)$.}\\
\kappa(\delta, \delta, k + \delta) &\in \BigO(\delta^{2\delta} \cdot (k + \delta)^{\delta^2(\delta+3)/2} \cdot 4^{\delta^2}) & \mbox{By Thm.~ \ref{simpleCutCharacteristicBound}.} \\
\kappa(\delta, \delta, k + \delta)^2 &\in \BigO(\delta^{4\delta} \cdot (k^{\delta^2(\delta+3)} + \delta^{\delta^2(\delta+3)}) \cdot 16^{\delta^2}) & \mbox{Simplifying.} \\
|Z|^2 &\in \BigO(\alpha^2 \delta^2 \kappa(\delta, \delta, k + \delta)^2) \\
|Z|^2 &\in \BigO(k^{14} \delta^4 \delta^{4\delta} \cdot (k^{\delta^2(\delta+3)} + \delta^{\delta^2(\delta+3)}) \cdot 16^{\delta^2}) \label{zsquareBound}
\end{align}
If we now treat~$w$ (and therefore~$\delta$) as a fixed constant, we find:
\begin{align}
|Z|^2 &\in \BigO(k^{\BigO(w^3)}) & \mbox{By \eqref{zsquareBound}.} \label{sloppyZSquareBound} \\
|V(G^*)| &\in \BigO(k^{O(w^3)}) & \mbox{By \eqref{vBound} and \eqref{sloppyZSquareBound}.} \nonumber
\end{align}
}
Since this shows that the size of a reduced instance is appropriately bounded, this concludes the proof.
\qed 
\end{proof}
}

\polykernelproof

\section{On approximating the deletion distance to a bipartite treewidth-$w$ graph}

\newcommand{\approximatingtheparameter}{

\begin{proposition}[\cite{RobertsonS86,RobertsonS04}] \label{boundedtwExcludesPlanarGraph}
Let~$w \geq 1$ be an integer. There is a finite set of graphs~$\F_{\tw(w)}$ containing at least one planar graph, such that for any graph~$G$ we have~$\tw(G) \leq w$ if and only if~$G$ excludes all graphs~$H \in \F_{\tw(w)}$ as a minor.
\end{proposition}

\begin{proof}
It is well-known that the treewidth of a graph does not increase when taking a minor. Hence the set of graphs~$\G_{\tw(w)}$ of treewidth at most~$w$ is minor-closed, and by the Graph Minor Theorem~\cite{RobertsonS04} there is a finite obstruction set~$\F$ such that for all graphs~$G$ we have~$G \in \G_{\tw(w)}$ if and only if~$G$ excludes all graphs~$H \in \F$ as a minor. Observe that the~$k \times k$ grid graph has treewidth~$k$~\cite{RobertsonS86} (we ignore the easy case that~$w \leq 1$). Therefore~$\G_{\tw(w)}$ does not contain the~$(w+1) \times (w+1)$ grid graph. Since~$\F$ is an obstruction set for~$\G_{\tw(w)}$ there must be a graph~$H' \in \F$ which is a minor of the~$(w+1) \times (w+1)$ grid. But since a grid graph is planar, and planarity is preserved when taking minors, the graph~$H'$ must be planar; this proves the claim.
\qed
\end{proof}

We use the following theorem from recent work by Fomin et al.~\cite[Theorem 3]{FominLMPS11}.
\begin{theorem} \label{approximateHittingMinors}
Let \F be an obstruction set containing a planar graph. Given a graph~$G$, in polynomial time we can find a subset~$S \subseteq V(G)$ such that~$G - S$ contains no element of \F as a minor and~$|S| \in \BigO(\opt \cdot \log^{3/2} \opt)$. Here~$\opt$ is the minimum size of such a set~$S$.
\end{theorem}
We also use Courcelle's theorem for graphs of bounded treewidth.

\begin{proposition}[\cite{ArnborgLS91,Bodlaender96,BoriePT92,Courcelle90,CourcelleM93}]
\label{courcellesTheorem}
Let~$\phi$ be a property that is expressible in Monadic Second Order Logic. For any fixed integer~$w \geq 1$, there is an algorithm that, given a graph~$G$ of treewidth at most~$w$ as input, finds a largest (alternatively, smallest) set~$S$ of vertices of~$G$ that satisfies~$\phi$ in time~$f(w, |\phi|) |V (G)|$.
\end{proposition}

\begin{lemma} \label{boundedTwOCT}
For every fixed value of~$w \geq 1$ there is a linear-time algorithm which given a graph~$G$ of treewidth at most~$w$ computes a minimum-size odd cycle transversal~$S \subseteq V(G)$.
\end{lemma}

\begin{proof}
We will apply Courcelle's theorem, using the fact that OCT can be expressed in Monadic Second Order Logic: if~$G$ is a graph then finding a minimum OCT is equivalent to minimizing the cardinality of a set~$S \subseteq V(G)$ which satisfies the following MSOL formula~$\phi(S)$:
\begin{align*}
\phi(S) :=& \exists X, Y \subseteq V(G): \forall v \in V(G): [v \in S \vee v \in X \vee v \in Y] \wedge \\ & \mbox{Independent}(X) \wedge \mbox{Independent}(Y) \\
\mbox{where}\\ 
\mbox{Independent}(Z) :=& \forall u, v \in Z: \neg \mbox{Adj}(u,v)
\end{align*}
The formula uses the fact that if~$S$ is an odd cycle transversal then~$G - S$ is bipartite and hence the vertices of~$G - S$ can be covered by two independent sets. Since the formula~$\phi(S)$ is fixed and does not depend on~$w$, the lemma now follows from~\proposref{courcellesTheorem}.
\qed
\end{proof}

\begin{lemma} \label{approximateBipartiteBoundedTWDeletion}
Let~$w \geq 1$ be a fixed integer. There is a polynomial-time algorithm which gets as input a graph~$G$, and computes a set~$X \subseteq V(G)$ such that~$G - X \in \bipartite \cap \gtw{w}$  with~$|X| \in \BigO(\opt \cdot \log^{3/2} \opt)$, where~$\opt$ is the minimum size of such a deletion set.
\end{lemma}

\begin{proof}
We introduce a little bit of terminology to simplify the proof. Let~$G$ be a graph. A subset~$S \subseteq V(G)$ is called a \emph{treewidth-$w$ deletion set} for~$G$ if~$G - S \in \gtw{w}$. If the graph~$G - S$ is also bipartite, then~$S$ is a \emph{treewidth-$w$ odd cycle transversal}.

Now fix an integer~$w \geq 1$: we will sketch the actions of the algorithm on input~$G$. Let~$\opt_{\tw(w)}$ denote the minimum cardinality of a treewidth-$w$ deletion set for~$G$, and let~$\opt_{\tw(w) \oct}$ be the minimum cardinality of a treewidth-$w$ odd cycle transversal.

By \proposref{boundedtwExcludesPlanarGraph} there is a finite set of graphs~$\F_{\tw(w)}$ such that for all~$S \subseteq V(G)$ we have~$\tw(G - S) \leq w$ if and only if~$G - S$ excludes all graphs of~$\F_{\tw(w)}$ as a minor, and~$\F_{\tw(w)}$ includes a planar graph. Hence we may apply \thmref{approximateHittingMinors} to compute in polynomial time a set~$S_{\tw(w)}$ such that~$G - S_{\tw(w)}$ excludes all graphs of~$\F_{\tw(w)}$ as a minor (and hence~$\tw(G - S_{\tw(w)}) \leq w$), and~$|S_{\tw(w)}| \in \BigO(\opt _{\tw(w)} \cdot \log^{3/2} \opt_{\tw(w)})$.

Observe that we have~$\opt_{\tw(w)} \leq \opt_{\tw(w) \oct}$ since a treewidth-$w$ odd cycle transversal must also be a treewidth-$w$ deletion set. Now consider the graph~$G - S_{\tw(w)}$ which has treewidth at most~$w$, but which might not be bipartite. We will compute our approximate treewidth-$w$ odd cycle transversal by taking the union of~$S_{\tw(w)}$ and a minimum odd cycle transversal of the graph~$G - S_{\tw(w)}$. Since~$G - S_{\tw(w)}$ has treewidth at most~$w$ and we take~$w$ to be a constant, \lemmaref{boundedTwOCT} shows we can compute an optimal odd cycle transversal~$S_{\oct}$ of~$G - S_{\tw(w)}$ in linear time. Observe that since a treewidth-$w$ odd cycle transversal cannot be smaller than an odd cycle transversal of a subgraph, we must have~$|S_{\oct}| \leq \opt_{\tw(w)\oct}$. Since the graph~$(G - S_{\tw(w)}) - S_{\oct}$ has treewidth at most~$w$ and is bipartite, we find that~$X := S_{\tw(w)} \cup S_{\oct}$ is a treewidth-$w$ odd cycle transversal of~$G$ and it follows from our earlier observations that this set satisfies the claimed size bound.
\qed
\end{proof}}


\approximatingtheparameter

\section{Omitted proofs and definitions of Section~\ref{section:lowerbounds}}

\newcommand{\crosscompositiondefs}{

\begin{definition}[\cite{BodlaenderJK11}]
\label{polyEquivalenceRelation}
An equivalence relation~\R on~$\Sigma^*$ is called a \emph{polynomial equivalence relation} if the following two conditions hold:
\begin{enumerate}
	\item There is an algorithm that given two strings~$x,y \in \Sigma^*$ decides whether~$x$ and~$y$ belong to the same equivalence class in~$(|x| + |y|)^{O(1)}$ time.
	\item For any finite set~$S \subseteq \Sigma^*$ the equivalence relation~\R partitions the elements of~$S$ into at most~$(\max _{x \in S} |x|)^{O(1)}$ classes.
\end{enumerate}
\end{definition}

\begin{definition}[\cite{BodlaenderJK11}]
\label{crossComposition}
Let~$L \subseteq \Sigma^*$ be a set and let~$Q \subseteq \Sigma^* \times \mathbb{N}$ be a parameterized problem. We say that~$L$ \emph{cross-composes} into~$Q$ if there is a polynomial equivalence relation~\R and an algorithm which, given~$t$ strings~$x_1, x_2, \ldots, x_t$ belonging to the same equivalence class of~\R, computes an instance~$(x^*,k^*) \in \Sigma^* \times \mathbb{N}$ in time polynomial in~$\sum _{i=1}^t |x_i|$ such that:
\begin{enumerate}
	\item~$(x^*, k^*) \in Q \Leftrightarrow x_i \in L$ for some~$1 \leq i \leq t$,
	\item~$k^*$ is bounded by a polynomial in~$\max _{i=1}^t |x_i|+\log t$.
\end{enumerate}
\end{definition}

We point out that all logarithms are base two. For ease of reading we let a sequence of~$R$ zeros be the binary expansion of~$2^R$ (this still gives one unique representation for all numbers from~$1$ to~$2^R$).

\begin{theorem}[\cite{BodlaenderJK11}, Corollary 10] \label{crossCompositionNoKernel}
If a set~$L$ is NP-hard under Karp reductions and~$L$ cross-composes into a parameterized problem~$Q$ then there is no polynomial kernel for~$Q$ unless \containment.
\end{theorem}
}

\crosscompositiondefs

The following Theorems~\ref{theorem:lowerbound:outerplanar},~\ref{theorem:lowerbound:cluster}, ~\ref{theorem:lowerbound:cocluster}, and~\ref{theorem:lowerbound:weightedoctbyvertexcover} together imply Theorem~\ref{theorem:lowerbounds}.

\newcommand{\lowerboundproofouterplanar}{
\begin{theorem}\label{theorem:lowerbound:outerplanar}
$(\outerplanar)$-\oct does not admit a polynomial kernelization unless \containment.
\end{theorem}
\begin{proof}
We prove the lower bound by a cross-composition from (unparameterized) \textsc{Vertex Cover}. An instance~$x$ of \textsc{Vertex Cover} consists of a graph~$G=(V,E)$ and an integer~$\ell$, asking whether~$G$ has a vertex cover of size at most~$\ell$. We define two instances~$(G_i=(V_i,E_i),\ell_i)$ and~$(G_j=(V_j,E_j),\ell_j)$ to be equivalent under a relation~\R if and only if~$|V_i|=|V_j|$,~$|E_i|=|E_j|$, and~$\ell_i=\ell_j<|V_i|$. For technical reasons, we let all ill-formed instances, i.e., those not encoding a graph~$G$ and an integer~$\ell$, be equivalent under~\R. Since those instances are trivially \no, they may equivalently be deleted from the input to the cross-composition and we will tacitly ignore them henceforth. Similarly, we let all instances~$(G,\ell)$ with~$\ell\geq|V(G)|$ form one equivalence class. However, since such an instance is trivially \yes, the output can be any dummy \yes-instance, if such an instance is in the input to the cross-composition; henceforth~$\ell<|V(G)|$. Clearly,~\R is a polynomial equivalence relation since instances of size at most~$N$ are partitioned into at most~$N^3+2$ classes, and equivalence can be checked in polynomial time.

Let~$x_1,\ldots,x_t$ be~$t$ instances of \textsc{Vertex Cover} that are equivalent under~\R. W.l.o.g.\ we assume~$t=2^R$ for some integer~$R$ (otherwise we could copy one instance sufficiently often, at most doubling the input size). Each instance~$x_i$ asks whether a graph~$G_i$ on~$n$ vertices and~$m$ edges has a vertex cover of size at most~$\ell$. We construct an instance of odd cycle transversal parameterized by a modulator from an outerplanar graph, by first constructing a graph~$G'$ by adding \emph{instance selectors}, a \emph{solution selector}, and \emph{edge checkers}:
\begin{itemize}
\item \textbf{Instance selectors:} An instance selector consists of~$R$ vertex-disjoint triangles~$T_1,\ldots,T_R$. Two vertices of each triangle are called~$0$-vertex respectively~$1$-vertex. Any odd cycle transversal for the instance selector must contain at least~$R$ vertices. We make a total of~$n$ copies of this construction, which gives a total cost of at least~$n\cdot R$ vertex deletions for handling all odd cycles of the instance selectors.
\item \textbf{Solution selector:} We start from a clique on~$n$ vertices, corresponding to the~$n$ vertices in each of the graphs~$G_i$, and subdivide each of its edges once. We obtain a bipartite graph with bipartitions of size~$n$ (the original vertices) and~$\binom{n}{2}$ vertices. The so obtained independent set of size~$n$ will encode the selection of~$n$ vertices. The other~$\binom{n}{2}$ vertices will serve to complete odd cycles.
\item \textbf{Edge checkers:} We add~$n$ copies of the following edge checker for each of the~$m$ edges of each of the~$t$ graphs. The construction will be outerplanar.

We start with a path of at least~$R$ vertices and with an even number of vertices (i.e.,~$R$ or~$R+1$ vertices). To the first~$R$ of these vertices we add a triangle with a pending vertex.  Let~$v$ be such a vertex on the path: then we add vertices~$a_v$,~$b_v$, and~$c_v$ with edges~$\{v,a_v\}$,~$\{v,b_v\}$,~$\{a_v,b_v\}$, and~$\{b_v,c_v\}$ (i.e., a triangle on the vertices~$v$,~$a_v$, and~$b_v$ with a pending vertex~$c_v$ at~$b_v$). We note that any odd cycle transversal for an edge checker has size at least~$R$ since it contains~$R$ vertex-disjoint triangles. This implies a total cost of at least~$n\cdot mt\cdot R$ vertex deletions for removing all odd cycles from all edge checkers (as there are~$n\cdot mt$ edge checkers in total).

An edge checker for an edge~$\{p,q\}$ of~$G_i$ is connected to the instance selectors as well as to~$p$ and~$q$ in the solution selector in the following way; it ensures that one of~$p$ or~$q$ must be chosen if the~$i$-th instance is chosen. Let~$v_1,\ldots,v_R$ be the~$R$ vertices on the path to which we added triangles. For~$j\in[R]$, if the~$j$th bit of the binary expansion of~$i$ is zero then we connect~$b_{v_j}$ and~$c_{v_j}$ to the~$0$-vertex of the~$j$th triangle~$T_j$ of each instance selector. Otherwise we connect~$b_{v_j}$ and~$c_{v_j}$ to the~$1$-vertex of~$T_j$. In both cases the vertices~$b_{v_j}$ and~$c_{v_j}$ form a triangle with a vertex of~$T_j$. Finally, we connect the first vertex of the path to~$p$ and the last vertex to~$q$ (this latter connection is arbitrary, the roles of~$p$ and~$q$ may be exchanged without harm).
\end{itemize}
The cross-composition generates an instance~$x'=(G',\ell',X)$, where~$\ell':=n\cdot R+n\cdot mt\cdot R+\ell$, and~$X$ contains all vertices of the instance and the solution selectors. Clearly,~$G'-X$ is an outerplanar graph, as it is a disjoint union of edge checkers, and the parameter value~$k'=|X|$ of~$x'$ is bounded by~$n\cdot 3R+n+\binom{n}{2}$, which is polynomial in~$\max_i|x_i|+\log t$. It is easily checkable that the construction of~$G'$ can be performed in time polynomial in~$\sum_i |x_i|$.

For correctness of the cross-composition we will show that~$x'$ is \yes if and only if at least one instance~$x_i$ is \yes. I.e., we establish that~$G'$ has an odd cycle transversal of size at most~$\ell'$ if and only if at least one of the graphs~$G_i$ has a vertex cover of size at most~$\ell$.

$\boldsymbol{(\Rightarrow)}$ We assume that some instance~$x_i$ is \yes, and we let~$S$ be a vertex cover for~$G_i$ of size at most~$\ell$. We define a set~$S'$ to serve as an odd cycle transversal of~$G'$. First, we include from the solution selector the~$\ell$ vertices that correspond to~$S$. Second, we add from each of the~$n$ instance selectors~$0$- and~$1$-vertices matching the complement of the binary expansion of~$i$: if the~$j$th bit of~$i$ is~$0$ then we add the~$1$-vertex of~$T_j$ to~$S'$ and otherwise we add the~$0$-vertex. In edge checkers for edges of~$G_i$ we add the~$b$-vertices from the triangles. For other edge checkers, say, for graphs~$G_{i'}$ with~$i'\neq i$, we pick a position~$j$ where the binary expansions of~$i$ and~$i'$ differ, and add~$v_j$ to~$S'$. For the other positions we add~$b_{v_{j'}}$ to~$S'$ for all~$j'\in[R]\setminus\{j\}$. Thus we pick a set~$S'$ of size at most~$n\cdot R+n\cdot mt\cdot R+\ell$.

Let us argue that~$S'$ is indeed an odd cycle transversal for~$G'$. First, we observe that there are no odd cycles in the instance selectors in~$G'-S'$, since we selected the~$0$ or the~$1$ vertex of each triangle. Second, let us consider the edge checkers:
\begin{itemize}
\item In edge checkers for~$G_i$ we added all~$b$-vertices of the triangles, hence those checkers are disconnected from the instance selectors in~$G'-S'$. Furthermore, for each of the corresponding edges the set~$S$ contains one of its endpoints, and hence we have added one of its endpoints in the solution selector to~$S'$. It can be easily checked that the remainders of those edge checkers are (caterpillar) trees attached to at most one of the endpoints (in the solution selector) of the corresponding edge.
\item In other edge checkers, for graphs~$G_{i'}$, we have also added all~$b$-vertices of the triangles except for one position, say~$j$, where the binary expansions of~$i$ and~$i'$ differ; there we added~$v_j$. Hence, also these edge checkers are disconnected from the instance selectors since our choice on the instance selector (in particular in position~$p$) is exactly opposite to~$i$ (and matching~$i'$).

By adding~$v_j$ to~$S'$ we have ensured that in~$G'-S'$ such an edge checker is split into two (caterpillar) trees each attached to at most one vertex in the solution selector.
\end{itemize}
Clearly, the solution selectors are already bipartite in~$G'$ (and we have checked for odd cycles via the edge checkers). Hence~$S'$ is indeed an odd cycle transversal of~$G'$, and~$x'$ is \yes.

$\boldsymbol{(\Leftarrow)}$ We assume that~$x'$ is \yes. Let~$S'$ be an odd cycle transversal of~$G'$ of size at most~$\ell'$. From each instance selector and from each edge checker,~$S'$ must contain~$R$ vertices. Therefore it contains at most~$\ell$ vertices from the solution selector. We define a set~$S$ to contain all vertices of~$S'$ that are in the independent set of size~$n$ (corresponding to vertices of the graphs~$G_i$). Furthermore, for each vertex of~$S'$ that is contained in the independent set of size~$\binom{n}{2}$ of the solution selector (i.e., the one obtained by subdividing all edges of the initial clique on~$n$ vertices), we arbitrarily include one of its two neighbors (instead of the subdivision on, say,~$\{p,q\}$, we include~$p$ or~$q$). Clearly, the size of~$S$ is at most~$\ell$.

Let us first see that there must be an instance selector where~$S'$ selects exactly~$R$ vertices (recall that the minimum is~$R$). The reason is that the total size of~$S'$ would otherwise be at least
\[
n\cdot(R+1)+n\cdot mt\cdot R=n\cdot R+n\cdot mt\cdot R+n>n\cdot R+n\cdot mt\cdot R+\ell=\ell',
\]
since~$\ell<n$. Similarly it can be seen that for each edge of any graph~$G_i$ there must be an edge checker where~$S'$ picked the minimum number of~$R$ vertices. For the following, let us focus on one instance selector and one edge checker per edge and instance where~$S'$ selected the minimum of~$R$ vertices.

Let us first consider the instance selector. Since there are~$R$ triangles in the selector,~$S'$ must contain one vertex of each. We choose an integer~$i$ via its binary expansion: for all~$j\in[R]$, if~$S'$ does not include the~$0$-vertex of triangle~$T_j$ of the selector, then we let the~$j$th bit be~$0$; otherwise we let it be~$1$.

We claim that~$S$ is a vertex cover of~$G_i$, i.e., that~$x_i$ is \yes. Let~$\{p,q\}$ be an edge of~$G_i$ and consider the edge checker (i.e., one where~$S'$ selected exactly~$R$ vertices) corresponding to this edge. By choice of~$i$, for each~$j\in[R]$ the corresponding vertex of triangle~$T_j$ in the edge checker is present in~$G'-S'$. Hence, for each~$j\in[R]$, there is a triangle formed by~$b_{v_j}$,~$c_{v_j}$, and a vertex of~$T_j$, but~$S'$ does not contain the latter. Hence~$S'$ must contain~$b_{v_j}$ or~$c_{v_j}$ (in fact it must be~$b_{v_j}$ on account of the triangle at~$v_j$). Therefore,~$S'$ cannot contain any of the vertices~$v_1,\ldots,v_R$. This implies that the path through these vertices together with~$p$,~$q$, and the subdivision of the former edge~$\{p,q\}$ would give an odd cycle, and the only free option is that~$S'$ contains~$p$,~$q$, or the subdividing vertex. In any case,~$S$ must contain~$p$ or~$q$, proving that it is indeed a vertex cover of~$G_i$.
\qed
\end{proof}
}

\lowerboundproofouterplanar

\newcommand{\lowerboundproofcluster}{
\begin{theorem}\label{theorem:lowerbound:cluster}
$(\cluster)$-\oct does not admit a polynomial kernelization unless \containment.
\end{theorem}

\begin{definition}[\cite{BodlaenderJK11}] \label{K4InABoxDef}
The~$K_4$-in-a-box graph \KB (see \imgref{K4InABox}) is the graph obtained from a complete graph on~$4$ vertices~$\{a,b,c,d\}$ by adding a new degree-$2$ vertex~$v$ for each pair~$\{a,b\}, \{b,c\}, \{c,d\}, \{d,a\}$ such that~$v$ is adjacent to both vertices of the pair. The vertices~$\{a,c\}$ are the~$0$-labeled terminals of the graph, and the vertices~$\{b,d\}$ are the~$1$-labeled terminals of the graph.
\end{definition}

\begin{figure}
\centering
\begin{tikzpicture}[scale=0.9,thick]
\draw (0,1) -- (0,0) -- (2,0) -- (2,2) -- (0,2) -- (0,1) -- (1,0) -- (2,1) -- (1,2) -- (0,1);
\draw (1,0) -- (1,2);
\draw (2,1) -- (0,1);

\drawvertex{(0,0)}
\drawvertex{(0,1)}
\drawvertex{(0,2)}
\drawvertex{(1,0)}
\drawvertex{(1,2)}
\drawvertex{(2,0)}
\drawvertex{(2,1)}
\drawvertex{(2,2)}

\draw (-0.3,1) node {$0$};
\draw (2.3,1) node {$0$};
\draw (1,-0.3) node {$1$};
\draw (1,2.3) node {$1$};

\end{tikzpicture}
\caption{\label{K4InABox} The~$K_4$-in-a-box graph \KB with labeled vertices.}
\end{figure}
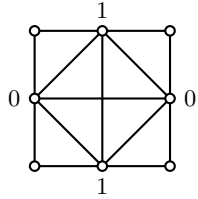

\begin{proof}[Theorem~\ref{theorem:lowerbound:cluster}]
We give a cross-composition from (unparameterized) \oct. An instance~$x$ of \oct consists of a graph~$G=(V,E)$ and an integer~$\ell$, asking whether~$G$ has an odd cycle transversal of size at most~$\ell$. 
We use the same equivalence relation~$\R$ as in the proof of Theorem~\ref{theorem:lowerbound:outerplanar}; w.l.o.g.~$\ell<|V|$.

Let~$x_1,\ldots,x_t$ be~$t$ instances of \oct that are equivalent under~\R. W.l.o.g.\ we assume~$t=2^R$ (otherwise we could copy one instance sufficiently often, at most doubling the input size). Each instance~$x_i$ asks whether a graph~$G_i$ on~$n$ vertices and~$m$ edges has an odd cycle transversal of size at most~$\ell$.

Consider a graph~$G_i$ before and after subdividing each edge with two vertices, i.e., with a path~$P_2$ of two vertices. It is easy to see that the subdivisions do not change whether or not~$G_i$ has an odd cycle transversal of size at most~$\ell$, since they do not change the parity of any path (i.e., the number of edges on any path is multiplied by three).
In a slight abuse of notation we use~$G_1,\ldots,G_t$ to denote the graphs obtained after the subdivision. Each consists of~$m$~$P_2$ as well as of an independent set of~$n$ vertices; we assume both the~$P_2$'s as well as the~$n$ independent vertices to be numbered from~$1$ to~$m$ and from~$1$ to~$n$, respectively (in each graph~$G_i$).

We will now construct a graph~$G'$ for the cross-composed instance, starting from a disjoint union of the graphs~$G_1,\ldots,G_t$:
\begin{itemize}
\item There are~$t$~$P_2$'s in~$G'$ for any given number~$i\in[m]$ at this point. We identify all of them to a single~$P_2$, for each~$i\in[m]$, and obtain~$m$~$P_2$'s.
\item We remark that all of the numbered vertices are only adjacent to~$P_2$'s at this point (due to the subdivisions). For each~$i\in[n]$ we turn all~$t$ vertices with number~$i$ into a clique by adding all possible edges between them. We add one universal vertex to each of these~$n$ cliques (i.e., a new vertex adjacent to all vertices of the clique). Thus, the independent set of all~$tn$ numbered vertices is turned into a disjoint union of~$n$ cliques each containing~$1$ vertex from each graph~$G_i$ plus one new vertex (i.e., size~$t+1$). Note that each~$G_i$ is still an induced subgraph of~$G'$.
\item For each~$p\in[R]$ we add~$n$ copies of the \KB graph. For~$\alpha\in\{0,1\}$, we connect the~$\alpha$-labeled terminals of each copy to all vertices of the cliques which correspond to a graph~$G_i$ such that the~$p$th position in the binary expansion of~$i$ is~$\alpha$.
\end{itemize}
We let~$x':=(G',\ell',X)$ denote the cross-composed instance. The set~$X$ denotes the set of all vertices in the~$m$~$P_2$'s as well as all vertices of the~$n\cdot R$ copies of the \KB graph. Clearly~$G'-X$ is a cluster graph since it consists only of the disjoint union of the~$n$ cliques. The size of~$X$, and hence the parameter value of~$x'$, is equal to~$2m+8nR$, i.e., polynomial in~$\max_i|x_i|+\log t$. The instance~$x'$ asks for an odd cycle transversal of~$G'$ of size at most~$\ell':=(t-1)\cdot n+2nR+\ell$.

We will now show that~$x'$ is \yes if and only if at least one of the instances~$x_i$ is \yes.

$\boldsymbol{(\Rightarrow)}$ Assuming that~$x_i$ is \yes, let~$S$ be an odd cycle transversal for~$G_i$ of size at most~$\ell$. We choose an odd cycle transversal~$S'$ of~$G'$. First, we add the at most~$\ell$ vertices of~$S$ (as~$G_i$ is a subgraph of~$G'$). Next, we add the~$(t-1)\cdot n$ vertices of the~$n$ cliques to~$S'$ which correspond to graphs~$G_{i'}$ with~$i'\in[t]\setminus\{i\}$. Then we add~$2$ vertices from each~\KB graph to~$S'$, matching the binary expansion of~$i$. The total size of~$S'$ is at most~$(t-1)\cdot n+2nR+\ell$.

We argue that~$G'-S'$ must be bipartite. Let us first consider the~\KB graphs in~$G'-S'$. In~$G'$ a~\KB graph has neighbors in the~$n$ cliques; they are connected either to its~$0$- or its~$1$-terminals. The set~$S'$ contains all those vertices, except for some that correspond to~$G_i$. However,~$S'$ was selected to contain exactly those terminals of the~\KB graphs that are adjacent to the vertices which correspond to~$G_i$. Hence, in~$G'-S'$ the remainders of the~\KB graphs form (bipartite) connected components of their own (bipartiteness after deletion of either~$0$- or~$1$-terminals can be easily checked).

Now let us consider the other components of~$G'-S'$. There is a copy of~$G_i-S$ and there are the vertices which were added to the~$n$ cliques as universal vertices ($1$ per clique), but the latter are adjacent to at most one vertex of~$G_i-S$, since they are only adjacent to vertices of their clique and~$S'$ deletes all vertices of other graphs~$G_j$ from~$G'$. Hence any~$2$-coloring of~$G_i-S$ can be easily extended to~$G'-S'$, implying that~$x'$ is a \yes-instance.

$\boldsymbol{(\Leftarrow)}$ Assuming that~$x'$ is \yes, let~$S'$ be an odd cycle transversal of~$G'$ of size at most~$\ell'=(t-1)\cdot n+2nR+\ell$. Clearly,~$S'$ must contain at least~$(t-1)$ vertices from each of the~$n$ cliques, since each of them contains~$t+1$ vertices. Similarly, it must contain at least~$2$ vertices of each~\KB graph, since it contains two vertex-disjoint triangles.

If~$S'$ would contain more than~$2$ vertices from all~$n$ copies of the~\KB graph corresponding to some position~$p\in[R]$ then the total size of~$S'$ would exceed~$\ell'$:
\[
|S'|\geq 3n+2n\cdot(R-1)+(t-1)\cdot n=2nR+(t-1)\cdot n+n>2nR+(t-1)\cdot n+\ell,
\]
since~$\ell<n$. Let us consider one~\KB and the selection of~$S'$ for each~$p\in[R]$. It can be easily checked that there are exactly two odd cycle transversals of size two for~\KB, namely choosing either the~$0$- or the~$1$-terminals. Let~$i\in[t]$ such that the~$p$ths position of its binary expansion matches the choice of terminals of~$S'$ in the corresponding~\KB. We claim that~$x_i$ is \yes.

Let~$v$ be a vertex in one of the cliques that corresponds to a graph~$G_{i'}$ with~$i'\in[t]\setminus\{i\}$. Let~$p\in[R]$ such that the binary expansions of~$i$ and~$i'$ differ in position~$p$. Thus there must be a~\KB graph corresponding to position~$p$ in which~$S'$ picked exactly the terminals that match~$i$, and therefore the other two terminals are present in~$G'-S'$. Since those two terminals form a triangle with~$v$ in~$G'$, we may conclude that~$S'$ contains~$v$.

Hence, in~$G'-S'$ there are no vertices left that correspond to graphs other than~$G_i$. Let us consider the induced copy of~$G_i$ in~$G'$ (as per construction) and the set of vertices~$S$ in which~$S'$ intersects it. Clearly, the induced copy does not contain the~\KB graphs and also does not contain the~$(t-1)\cdot n$ vertices of other graphs~$G_{i'}$. Therefore, since~$S'$ intersects those other parts of~$G'$ in a total of at least~$(t-1)\cdot n+2nR$ vertices, the set~$S$ contains at most~$\ell$ vertices. Since~$G'-S'$ is bipartite, the same must be true for~$G_i-S$, which implies that~$x_i$ is \yes.
\qed
\end{proof}}

\lowerboundproofcluster

\newcommand{\lowerboundproofcocluster}{
\begin{theorem}\label{theorem:lowerbound:cocluster}
$(\cocluster)$-\oct does not admit a polynomial kernelization unless \containment.
\end{theorem}

\begin{proof}
We give a cross-composition from (unparameterized) \oct. An instance~$x$ of \oct consists of a graph~$G=(V,E)$ and an integer~$\ell$, asking whether~$G$ has an odd cycle transversal of size at most~$\ell$. 
We use essentially the same equivalence relation~$\R$ as in the proof of Theorem~\ref{theorem:lowerbound:outerplanar}, except that w.l.o.g.~$\ell<|V|-2$ (since instances with~$\ell\geq|V|-2$ are trivially \yes).

Let~$x_1,\ldots,x_t$ be~$t$ instances of \oct that are equivalent under~\R. 
Each instance~$x_i$ asks whether a graph~$G_i$ on~$n$ vertices and~$m$ edges has an odd cycle transversal of size at most~$\ell$. Again we assume that~$t=2^R$.

Consider a graph~$G_i$ before and after subdividing each edge with two vertices, i.e., with a path~$P_2$ of two vertices. It is easy to see that the subdivisions do not change whether or not~$G_i$ has an odd cycle transversal of size at most~$\ell$, since they do not change the parity of any path (i.e., the number of edges on any path is multiplied by three).
In a slight abuse of notation we use~$G_1,\ldots,G_t$ to denote the graphs obtained after the subdivision. Each consists of~$m$~$P_2$ as well as of an independent set of~$n$ vertices; we assume both the~$P_2$'s as well as the~$n$ independent vertices to be numbered from~$1$ to~$m$ and from~$1$ to~$n$, respectively (in each graph~$G_i$). For~$i\in[t]$ we let~$\I_i$ denote the independent set on the~$n$ numbered vertices of~$G_i$. 

We will now construct a graph~$G'$ for the cross-composed instance, starting from a disjoint union of the graphs~$G_1,\ldots,G_t$:
\begin{itemize}
\item There are~$t$~$P_2$'s in~$G'$ for any given number~$i\in[m]$ at this point. We identify all of them to a single~$P_2$, for each~$i\in[m]$, and obtain~$m$~$P_2$'s.
\item It can be easily seen that~$\bigcup_i\I_i$ is an independent set in~$G'$. We add all edges~$\{u,v\}$ for any~$u\in\I_i$ and~$v\in\I_j$ with~$i\neq j$. Thus~$G'$ now contains~$\I_1\oplus\ldots\oplus\I_t$, i.e., the join of the~$t$ independent sets. Note that each~$G_i$ is still an induced subgraph of~$G'$.
\item We add vertices~$v_{p,i,j}$, with~$p\in[R]$ and~$i,j\in[n]$, which we connect to the independent sets~$\I_1,\ldots,\I_t$ as follows:
\begin{itemize}
\item We add an edge to the~$i$th vertex of the~$r$th independent set if the~$p$th bit in the binary expansion of~$r$ is zero.
\item We add an edge to the~$j$th vertex of the~$r$th independent set if the~$p$th bit in the binary expansion of~$r$ is one.
\end{itemize}
We note that each such vertex~$v_{p,i,j}$ has exactly~$t$ neighbors: one in each independent set of the co-cluster.

We make a total of~$2n$ copies of~$v_{p,i,j}$ for each~$p\in[R]$ and~$i,j\in[n]$.
\end{itemize}
Let~$X$ be a subset of the vertices of~$G'$ containing the~$2m$ vertices of the~$m$~$P_2$'s as well as the~$2n\cdot Rn^2$ copies of vertices~$v_{p,i,j}$. Clearly,~$G'-X$ is a co-cluster since it only contains the join of the~$t$ independent sets. We let~$\ell':=(t-1)\cdot n+\ell$ and define the cross-composed instance as~$x':=(G',\ell',X)$. It is easy to see that the parameter value, i.e., the size of~$X$, is polynomial in~$\max_i|x_i|+\log t$ and that the construction can be performed in polynomial time.

For correctness of the cross-composition we will now show that~$(G',\ell',X)$ is \yes if and only if one of the instances~$(G_i,\ell)$ is \yes:

$\boldsymbol{(\Leftarrow)}$ Let~$i\in[t]$ such that~$(G_i,\ell)$ is \yes and let~$S$ be an odd cycle transversal for~$G_i$ of size at most~$\ell$. To get an odd cycle transversal~$S'$ for~$G'$ we add to~$S$ all vertices of the independent sets~$\I_j$ for~$j\neq i$. The graph remaining after deletion of~$\bigcup_{j\neq i}\I_j$ contains a copy of~$G_i$ plus the vertices~$v_{p,i,j}$ but the latter have only one neighbor in~$G_i$ (since they have exactly one per independent set). Thus any 2-coloring of~$G_i-S$ can be extended to~$G'-S'$, and hence~$S'=S\cup\bigcup_{j\neq i}\I_j$ is an odd cycle transversal of~$G'$ of size at most~$(t-1)\cdot n+\ell$. This implies that~$x'=(G',\ell',X)$ is \yes too.

$\boldsymbol{(\Rightarrow)}$ Let~$S'$ be an odd cycle transversal of~$G'$ of size at most~$\ell'=(t-1)\cdot n+\ell$. We first observe that~$S'$ must contain all vertices of all but at most two of the independent sets~$\I_i$, since any three vertices from different independent sets induce a triangle. Therefore,~$S'$ cannot contain all~$2n$ copies of any~$v_{p,i,j}$ vertex, since then its total size would be at least~$t\cdot n>(t-1)\cdot n+\ell$; taking into account the~$2n$ copies and the at least~$(t-2)\cdot n$ vertices of the~$(t-2)$ independent sets that it contains.

Now, consider any two vertices say~$u$ and~$v$ from different independent sets, say~$\I_r$ and~$\I_s$, with~$r\neq s$. Let~$p$ be a bit position where~$r$ and~$s$ differ. Hence there are integers~$i$ and~$j$ such that~$v_{p,i,j}$ vertices are adjacent to~$u$ and~$v$ in~$G'$. Since~$S'$ cannot contain all~$v_{p,i,j}$ vertices, it must contain at least one of~$u$ and~$v$; as~$u$,~$v$, and~$v_{p,i,j}$ induce a triangle. Thus~$S'$ must contain all vertices from at least~$(t-1)$ independent sets, say from all but~$\I_j$. Restricting~$S'$ to~$\I_j$ and the~$P_2$-vertices we must obtain an odd cycle transversal~$S$ for (the subdivided version of)~$G_j$. Clearly,~$S$ is of size at most~$\ell$ since we remove at least~$(t-1)\cdot n$ vertices from~$S'$ to get it. Thus~$(G_j,\ell)$ is a \yes-instance.
\qed
\end{proof}}

\lowerboundproofcocluster

\newcommand{\lowerboundproofweightedoctbyvertexcover}{
\begin{theorem}\label{theorem:lowerbound:weightedoctbyvertexcover}
\textsc{Weighted Odd Cycle Transversal Parameterized by the size of a Vertex Cover} does not admit a polynomial kernelization unless \containment.
\end{theorem}

\begin{proof}[sketch.]
The proof is similar to the one for \thmref{theorem:lowerbound:cocluster}. We will only sketch the construction and mention the main idea of how the weights are used. As for \thmref{theorem:lowerbound:cocluster} the proof goes by cross-composition from~$\oct$. We use the same polynomial equivalence relation, so let the input for the cross-composition be instances~$(G_1,\ell),\ldots,(G_t,\ell)$ where each~$G_i$ is a graph with~$n$ vertices and~$m$ edges. Further, for ease of presentation, we assume all graphs to already have passed the subdivision step, where each edge is subdivided by a~$P_2$ (see the proof of \thmref{theorem:lowerbound:cocluster}). Finally w.l.o.g.~$t$ is a power of two.

The construction is as follows, starting from a graph~$G'$ that is a disjoint union of the graphs~$G_1$ through~$G_t$. Recall that each graph~$G_i$ consists of an independent set, denoted by~$I_i$, as well as~$m$ non-adjacent~$P_2$'s, numbered arbitrarily from~$1$ to~$m$.
\begin{itemize}
\item First, we identify all~$P_2$'s of the same number into one. We retain a graph with the~$t$ independent sets~$I_1,\ldots,I_t$ as well as~$m$ non-adjacent~$P_2$'s. Note that all the information about the graphs~$G_i$ now lies in the adjacency of the independent sets to the~$P_2$'s.
We also observe that each~$G_i$ (with the~$P_2$ subdivisions made) is an induced subgraph of our current graph~$G'$ (the identification does not change that). Thus, the final piece is an instance selector which forces the deletion of all but one independent set.
\item Recall the graph~$\KB$ (see Definition~\ref{K4InABoxDef}) with the two pairs of vertices labeled~$0$ and~$1$ respectively. We add~$\log t$ copies of it to our graph and connect the labeled vertices to all independent set vertices according to the binary expansions of their numbers (each~$\KB$ corresponds to one of the~$\log t$ bit positions). Recall that in each~$\KB$ already two vertex deletions are necessary to remove all odd cycles (delete the~$0$ or the~$1$ labeled vertices).
\end{itemize}
Giving the vertices of the~$\KB$ graphs a very high weight~$w$ and allowing a budget of~$\ell'=2w\log t+(t-1)n+\ell$ we force that in any solution of cost at most~$\ell'$ exactly~$2$ vertices can be deleted in each~$\KB$; the cost is~$2w\log t$. By construction there will be exactly one independent set whose number~$i\in[t]$ is such that it is disconnected from the remains of all~$\KB$ graphs, i.e., such that exactly the adjacent labeled vertices in the~$\KB$ graphs were deleted. All vertices of other independent sets are adjacent to two labeled vertices in at least one~$\KB$ graph which gives rise to a triangle. As the budget prohibits the deletion of more labeled vertices, all those vertices of the independent sets must be deleted; this costs~$(t-1)n$. What is left is a copy of one graph~$G_i$ with subdivided edges plus the (disconnected) bipartite remainders of the~$\KB$ graphs. Thus the remaining budget of~$\ell$ must go into an odd cycle transversal for that graph~$G_i$. Clearly, making~$w$ larger than~$(t-1)n+\ell$, e.g.,~$w=tn$ is sufficient for the above to work. It is straightforward to construct an odd cycle transversal for~$G'$ at cost at most~$\ell'$ given a transversal of size at most~$\ell$ for one graph~$G_i$. Finally, the set~$X$ defined to contain the~$2m$ vertices of the (identified)~$P_2$'s as well as the~$8\log t$ vertices of the~$\KB$ graphs can be given as a vertex cover for~$G'$; its size is bounded by a polynomial in the maximum instance size (larger than~$n+m$) and~$\log t$, as required. Clearly, polynomial time is enough to perform the construction.
\qed
\end{proof}}

\lowerboundproofweightedoctbyvertexcover

\section{Bibliography for appendix}
In this section we list the bibliographic information for items which were referenced in the appendix. Since these will not be present in the camera-ready version, we do not want to sacrifice space for these references in the alloted 12 pages.

\putbib[../Paper]
\end{bibunit}

\end{document}